\documentclass[10pt,journal,compsoc,dvipsnames]{IEEEtran}

\ifCLASSOPTIONcompsoc
  \usepackage[nocompress]{cite}
\else
  \usepackage{cite}
\fi
\ifCLASSINFOpdf
\else
\fi
\ifCLASSOPTIONcompsoc
 \usepackage[caption=false,font=footnotesize,labelfont=sf,textfont=sf]{subfig}
\else
 \usepackage[caption=false,font=footnotesize]{subfig}
\fi
\newcommand\MYhyperrefoptions{hidelinks,bookmarks=true,bookmarksnumbered=true,
pdfpagemode={UseOutlines},plainpages=false,pdfpagelabels=true,
colorlinks=true,linkcolor={black},citecolor={black},urlcolor={black},
pdftitle={E-Tenon: An Efficient Privacy-Preserving Secure Open Data Sharing Scheme for EHR System},
pdfsubject={Typesetting},
pdfauthor={Michael D. Shell},
pdfkeywords={Computer Society, IEEEtran, journal, LaTeX, paper,
             template}}
\ifCLASSINFOpdf
\usepackage[\MYhyperrefoptions,pdftex]{hyperref}
\else
\usepackage[\MYhyperrefoptions,breaklinks=true,dvips]{hyperref}
\usepackage{breakurl}
\fi
\usepackage{orcidlink}
\usepackage{tikz} 
\usepackage{amsthm, amsmath, amssymb, wasysym, pifont} 
\usepackage{multirow} 
\usepackage{enumitem} 
\usepackage{algorithm}  
\usepackage{algpseudocode} 
\usepackage{pgfplots}
\pgfplotsset{compat=1.14}

\newtheorem{definition}{Definition}
\newtheorem{theorem}{Theorem}

\usetikzlibrary{fpu, positioning, shapes}

\graphicspath{ {figures/} } 

\begin{document}
%
\title{E-Tenon: An Efficient Privacy-Preserving Secure Open Data Sharing Scheme for EHR System}
%
%
%
%

\author{Zhihui~Lin,
  Prosanta~Gope$^{\orcidlink{0000-0003-2786-0273}}$,~\IEEEmembership{Senior~Member,~IEEE,}
  Jianting~Ning$^{\orcidlink{0000-0001-7165-398X}}$,~\IEEEmembership{Member,~IEEE,}
  and ~Biplab~Sikdar$^{\orcidlink{0000-0002-0084-4647}}$,~\IEEEmembership{Senior~Member,~IEEE,}
  \IEEEcompsocitemizethanks{
    \IEEEcompsocthanksitem Zhihui Lin and Prosanta Gope are with the Department of Computer Science, University of Sheffield, S10 2TN Sheffield, U.K. \protect\\(e-mail: \href{mailto:zhihuilin111@gmail.com}{zhihuilin111@gmail.com}; \href{mailto:p.gope@sheffield.ac.uk}{p.gope@sheffield.ac.uk}).
    \IEEEcompsocthanksitem Jianting Ning is with College of Computer and Cyber Security, Fujian Normal University, Fuzhou 350117, China (e-mail: \href{mailto:jtning88@gmail.com}{jtning88@gmail.com}).
    \IEEEcompsocthanksitem Biplab Sikdar is with Department of Electrical and Computer Engineering,\\ National University of Singapore, Singapore 117583, Singapore (e-mail: \href{mailto:bsikdar@nus.edu.sg}{bsikdar@nus.edu.sg}).}
  \thanks{}}

%
%

\markboth{}%
{Shell \MakeLowercase{\textit{et al.}}: Bare Advanced Demo of IEEEtran.cls for IEEE Computer Society Journals}
%



\IEEEtitleabstractindextext{%
  \begin{abstract}
    The transition from paper-based information to Electronic-Health-Records (EHRs) has driven various advancements in the modern healthcare-industry. In many cases, patients need to share their EHR with healthcare professionals. Given the sensitive and security-critical nature of EHRs, it is essential to consider the security and privacy issues of storing and sharing EHR. However, existing security solutions excessively encrypt the whole database, thus requiring the entire database to be decrypted for each access request, which is a time-consuming process. On the other hand, the use of EHR for medical research (e.g., development of precision-medicine, diagnostics-techniques), as well as optimisation of practices in healthcare organisations, requires the EHR to be analysed, and for that, they should be easily accessible without compromising the privacy of the patient. In this paper, we propose an efficient technique called E-Tenon that not only securely keeps all EHR publicly accessible but also provides the desirable security features. To the best of our knowledge, this is the \emph{first work} in which an Open Database is used for protecting EHR. The proposed E-Tenon empowers patients to securely share their EHR under multi-level, fine-grained access policies defined by themselves. Analyses show that our system outperforms existing solutions in terms of computational-complexity.
  \end{abstract}

  \begin{IEEEkeywords}
    open database, e-tenon, ABE, multi-level attribute-based encryption, multi-signature.
  \end{IEEEkeywords}}

\maketitle

\IEEEdisplaynontitleabstractindextext

%
\IEEEpeerreviewmaketitle

\IEEEraisesectionheading{\section{Introduction}\label{section:1}}
\IEEEPARstart{W}{ith} the rapid development of Health Information Technology (HIT) and cloud services, a growing number of healthcare organisations are accelerating the implementation of Electronic Health Record (EHR) based systems. These systems enhance their services and core competencies since EHRs can address many limitations of traditional paper-based medical records, such as scalability, accessibility, and persistence. EHRs are often shared across doctors and healthcare providers with patients' consent and typically include a range of sensitive and private information such as patient's identity codes, health history, medical diagnoses and treatment plans. Obviously, leakage of these data can cause embarrassment or even result in life-threatening consequences for patients. Indeed, in reality, despite record levels of security spending by different hospitals, there is still a wide range of malicious cyberattacks intended to penetrate databases and connected systems. This is because cybercriminals find EHRs highly profitable, which motivates them to steal such data by various means. Therefore, designing a system that preserves patient privacy in a robust and efficient manner is imperative. 

\textbf{Motivation:} for the application scenarios mentioned above, many existing schemes are vulnerable and ineffective. For example, a common approach recommended and practised in the industry by many security practitioners is to strictly encrypt databases so that the data is protected to the maximum extent possible, even in the event of a security incident. However, it is worth noting that recent trends and incidents such as the COVID-19 outbreak has caused a sharp increase in the volume of medical information held by hospitals. Therefore, it is inefficient to excessively encrypt the whole database or a majority of the data as it will have a marked impact on the performance of the EHR system. Besides, EHRs are increasingly being used for developing customised and precision medicine regimens, developing new and more accurate techniques for diagnosis and treatment, and optimising medical processes to help healthcare organisations to meet growing medical demands, improve operations, and reduce costs. Such applications require the EHRs to be easily accessible for analysis without compromising privacy.

A naive solution to the above requirements is to use Attribute-Based Encryption (ABE), which provides confidentiality and fine-grained access control. There are two general types of ABE: Ciphertext-Policy Attribute-Based Encryption (CP-ABE) \cite{Bethencourt2007} and Key-Policy Attribute-Based Encryption (KP-ABE) \cite{Goyal2006}. In CP-ABE, data is encrypted with a user-defined access structure, and a user with the relevant attributes can decrypt it \cite{Kaaniche2019}. Contrarily to CP-ABE, KP-ABE encrypts data with a set of descriptive attributes, and a user with a key embedded with an appropriate access structure can decrypt the data \cite{Kaaniche2019}. In this paper, we focus on CP-ABE. As an example, suppose a data owner (patient) wants to share \textit{part} of their EHR with a specific healthcare professional with a specific role and responsibility according to a CP-ABE access policy defined by the patient. One doctor is assigned a specific set of attributes (e.g., \{"doctor", "temporary", "oncology", "top secret"\}) while another doctor may be assigned different attributes (e.g., \{"doctor", "in charge", "dentistry", "confidential"\}). 

A doctor is authorised to access a patient's EHR data if his/her attribute set satisfies the access policy for that \textit{part} of the data. Readers may have noticed that naive CP-ABE does not perfectly support multi-level access control, meaning that data owners have to individually set different access policies for different parts of EHRs, depending on the type and sensitivity. As a result, these access policies will introduce many \textit{duplicate attributes}. The number of access policies is proportional to the number of duplicate attributes. That is, as more access policies are registered in the system, there will be more duplicate attributes, which will undoubtedly increase the storage and communication overheads. Worse still, given that most EHR data is sensitive, encrypting and decrypting large volumes of sensitive EHR using naive CP-ABE can be prohibitively expensive (i.e., it suffers from \textit{linear decryption cost} \cite{8007294}), especially for resource-constrained devices. This, however, yields the following questions: (1) Can we retain the benefits of CP-ABE for fine-grained access control while avoiding duplication of attributes when implementing multi-level access control in EHR systems? (2) Can we protect the confidentiality of EHR data without relying on extensive encryption (i.e., \textit{securely keep most of the EHR data open/in plaintext})?

Another concern is related to the integrity and authenticity of EHRs. Apart from the information provided by the medical staff, nowadays, more and more health data are collected from connected sensors, and wearable medical devices \cite{Gope2016_BSN} over (insecure) networks. All of this is patient-centred data, over which \textit{the patient has primary control}. However: (1) what if the patient takes advantage of his or her primary control to share/upload false data to the EHR system? (2) What if a medical device exploits its automated nature to upload false data to the EHR system without the patient's consent? (3) What if a man-in-the-middle intercepts and tampers with the data? 

\subsection{Our Contributions}
Indeed, strict security requirements appear to be diametrical to the goal of keeping data open. This paper makes a novel attempt to address these seemingly contradicting requirements. It proposes a novel E-Tenon system where data are stored in an open database while maintaining all privacy and security properties. 

One of the core components of the proposed system is the Tenon database (TDB), whose overview is presented in Fig. \hyperref[fig:tdb-general]{1}. Unlike conventional databases, the TDB is an open database consisting of a series of public tables and one secret table. Its main advantage lies in the fact that data protection does not depend on heavy encryption and decryption. Instead, the protection of EHRs is achieved through data preprocessing, maintenance of secret relationships between EHR blocks, and shuffling techniques. Notably, EHRs will be classified into identifiable information and Non Personally Identifiable Information (Non-PII), the latter of which will be tokenised into EHR blocks and can be securely made public. In addition, EHRs in the TDB are constantly shuffled, which makes it extremely difficult for attackers to exploit the open data. The main contributions of this paper are summarised as follows:
\begin{figure}[t]
    \centering
    \includegraphics[width=\columnwidth]{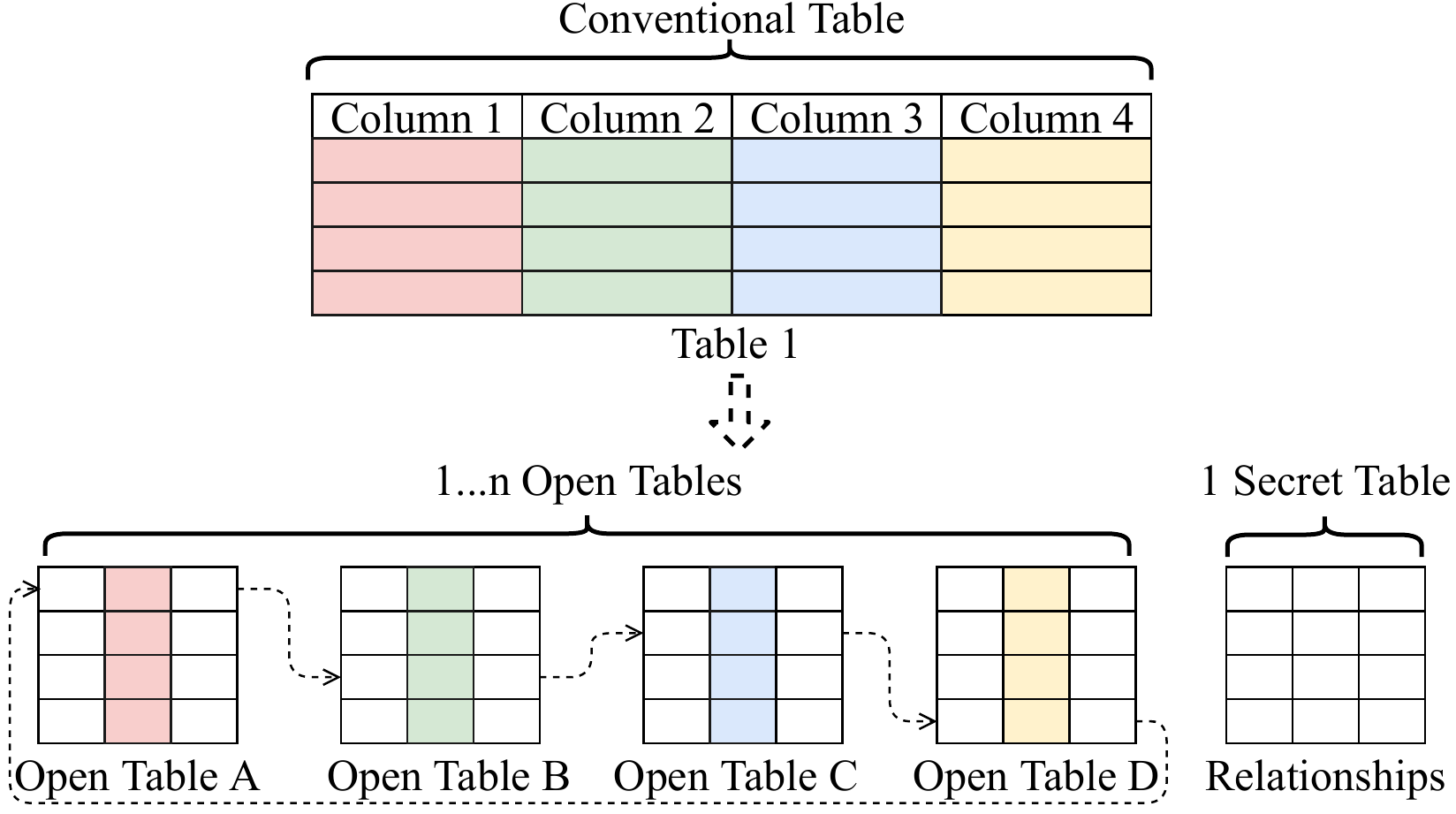}
    \caption{Overview of the proposed tenon database. A conventional table will be segmented into a series sub-tables where the relationship between rows are hidden. It can be revealed partially or fully, depending on the data user's attributes (access rights).}
    \label{fig:tdb-general}
\end{figure}
\begin{itemize}
\item We design an efficient open database in which most of the data are open, and only a minor portion needs to be encrypted. Thus, it requires less computation than other schemes in terms of encryption and decryption.
\item We integrate and extend Multi-Signature and Multi-level Attribute-based Encryption techniques to satisfy all desired security properties.
\item We present data preprocessing and shuffling methods used in conjunction with the proposed E-Tenon system to store and share EHRs securely in an open database setting.
\item We show how to ensure that a medical device and a data owner sign the same content, even if the EHRs have been preprocessed. This guarantees the authenticity and integrity of EHRs.
\end{itemize}

Our work addresses the shortcomings of previous solutions since E-Tenon not only efficiently guarantees multi-level, fine-grained EHR-data sharing but also protects the integrity and authenticity of the EHR, most importantly, \textit{under an open database setting}. It takes only 2.34 milliseconds for signing and verifying the signature, and 0.14 and 0.76 seconds for encryption and decryption of the secret relationship, respectively. To the best of our knowledge, E-Tenon is the first open database-based scheme to provide such a wide range of security and privacy properties. Note that while the focus of this work is on EHR, the concept of E-Tenon would also be applicable in other scenarios which require low-latency access to user data, such as in mobile edge computing environments. 

\subsection{Related Work}
Since medical data security has become a growing public concern, a considerable number of schemes have been published for secure medical data sharing and privacy preservation~\cite{8007294,Rezaeibagha2016_EHR,Bahga2013_EHR,Liu2019,9609621,8713588,shi2014authorized,li2012scalable_EHR,9520256,belguith2018phoabe_EHR,Sun2018_EHR,Zhang2018_EHR,Kumari2020_EHR}. For instance, most research in protecting medical data have emphasised the use of cryptographic methods such as CP-ABE and KP-ABE~\cite{8713588,9520256,9609621,belguith2018phoabe_EHR,Zhang2018_EHR,Green2011}. The system architecture proposed in~\cite{Rezaeibagha2016_EHR} is based on a successor of CP-ABE and Role-Based Access Control (RBAC) to protect EHR stored in the hybrid cloud with direct and indirect access. In Li et al.'s KP-ABE based model~\cite{li2012scalable_EHR}, the data owner needs to trust the key issuer because they are only inserting a set of descriptive attributes into the data using KP-ABE, but they do not know who will be accessing their data \cite{Bethencourt2007}. Xu et al. \cite{9609621} presented a practical dual-policy ABE scheme for EHR systems that combines the advantages of CP-ABE and KP-ABE with support for user revocation. In addition, Belguith et al.~\cite{belguith2018phoabe_EHR} proposed a multi-authority CP-ABE scheme that delegates expensive computing tasks to cloud servers, and their scheme also prevents collusion between the authorities. 

Nevertheless, it is unfortunate that most of the existing solutions in the literature are designed for private database settings, and there is no work which is able to ensure EHR security and patient privacy while keeping data in plaintext form. When every second matters during an emergency, the time-consuming encryption and decryption operations in a healthcare information system may cause delays in accessing patient information (such as medical history) during the golden hour that saves a patient's life. Likewise, as argued in~\cite{Dagher2018}, excessive security may obstruct sensible data use by healthcare providers and patients, and most approaches have failed to properly weigh the patients' right to privacy against the legitimate sharing of data.

Although several similar works mentioned above have used ABE to protect EHR, which is promising for flexible and fine-grained EHR sharing, they are computationally intensive when applied to encrypt the entire database. In addition, most solutions cannot support searching over encrypted data directly. Consequently, to search for relevant patient data in an encrypted database, the system first needs to decrypt the data on the application back-end. Such a burdensome process wastes valuable computing resources. Furthermore, many schemes fail to use digital signatures to ensure data integrity and authenticity properly. For example, \cite{Zhang2018_EHR} allows only one entity to sign the EHR, which grants the entity too much power. Despite some schemes \cite{Wang2018_EHR} allowing multiple entities to sign the data, they cannot guarantee that the same content is being signed honestly by all participants.

In general, to our knowledge, no state-of-the-art work on sharing and protecting EHRs has considered using a secure open database to save the avoidable overhead of encryption and decryption. That said, as the current solutions are built on private databases by default, we are unable to find related work that fully meets our expectations.

\subsection{Organisation}
The rest of the paper is organised as follows. Section \hyperref[section:2]{2} introduces and recapitulates the required mathematical notations, security assumptions and related schemes. In Section \hyperref[section:3]{3}, we present the system model and the corresponding adversarial model. This is followed by the construction of E-Tenon, given in detail in Section \hyperref[section:4]{4}. Next, we prove the security and practicality of the proposed scheme by conducting security and performance analysis in Sections \hyperref[section:5]{5} and \hyperref[section:6]{6}, respectively. Section \hyperref[section:7]{7} of the paper concludes our work in light of all that has been mentioned.

\section{Preliminaries}\label{section:2}
This section introduces and recapitulates several prerequisites, including definitions of some mathematical notations, a multi-level ABE scheme, and a multi-signature scheme.

\subsection{Notations}
We use $r\xleftarrow{\$}\mathbb{R}$ to mean that $r$ is chosen at random from $\mathbb{R}$, and $o\leftarrow A(i_{1},i_{2},\cdots,i_{n})$ to denote an algorithm $A$ that takes $i_{1}$ to $i_{n}$ as input parameters and yields the outcome of its operation $o$. If an algorithm returns $\bot$, it symbolises that the algorithm has failed to perform the expected actions ($v(\bot)=False$). $\mathbb{Z}_p$ is the set of integers modulo $p$, such that $\mathbb{Z}_{p} =\{[ 0]_{p} ,[ 1]_{p} ,\cdots,[ p-1]_{p}\}$. $\mathbb{G}$ is a multiplicative group of prime order \textit{p} where $0\notin \mathbb{G}$ since the multiplicative inverse of 0 does not exist. In addition, we denote $\mathbb{G}\backslash \{1\}$ by $\mathbb{G}^*$.

\subsection{Building Blocks}
More formal definitions are provided below. Bilinear maps are a useful tool for pairing-based cryptography because they conveniently establish relationships between cryptographic groups. As cyclic groups are used in the bilinear map, we first introduce the definition of a cyclic group.

\begin{definition}[Cyclic Group of Prime Order~\cite{Schnorr1991,Bellare2006}]
Let $\mathbb{G}_0 =\langle g\rangle$ be a cyclic group of prime order p where $\langle g\rangle =\left\{g^{n} :n\in \mathbb{Z}\right\}$, generator $g\in \mathbb{G}_0$, and p is a k-bit integer. Note that $\mathbb{G}_0$ can be denoted multiplicatively, and $\langle g\rangle$ is a cyclic subgroup of $\mathbb{G}_0$ generated by g.
\end{definition}

\begin{definition}[Prime Order Bilinear Group~\cite{Bethencourt2007}]
Let $\mathbb{G}_0$ and $\mathbb{G}_1$ be two multiplicative cyclic groups of same prime order p. g is an arbitrary generator $g\xleftarrow{\$}\mathbb{G}_0$. e is a symmetric bilinear map, such that $e:\mathbb{G}_{0} \times \mathbb{G}_{0}\rightarrow \mathbb{G}_{1}$ where $e\left( g^{x} ,g^{y}\right) =e\left( g^{y} ,g^{x}\right) =e( g,g)^{xy} =e( g,g)^{yx}$.
\end{definition}
There are three properties of an efficiently-computable \textit{e} that are worth noting:
\begin{enumerate}
\item \textit{Bilinearity}: $e\left( g^{y} ,g^{x}\right)$ and $e(g,h)^{xy}$ must be equivalent for all $x,y\xleftarrow{\$}\mathbb{Z}_p$ and $g_i,g_j\xleftarrow{\$}\mathbb{G}_0$.
\item \textit{Non-degeneracy}: $e(g,g)$ must not be equal to the identity of $\mathbb{G}_1$.
\item \textit{Computability}: for all $g,h\xleftarrow{\$}\mathbb{G}_0$, there exists an algorithm that can efficiently compute $e(g ,h)$.
\end{enumerate}

\begin{definition}[Discrete Logarithm Assumption~\cite{Bellare2006}]
Let $\mathbb{G}_0$ be a multiplicative cyclic group with a prime order $p$ and a generator $g$. The advantage is formulated as follows when a Probabilistic Polynomial-Time algorithm $\mathcal{A}$ is applied to solve the discrete logarithmic problem in $\mathbb{G}_0$:
\begin{gather*}
\mathrm{Adv}^{dlog}_{\mathbb{G}_{0}}( A) \!=\! \mathrm{Pr}\left[g^{x} =y|g\xleftarrow{\$}\mathbb{G}^{*}_{0};y\xleftarrow{\$}\mathbb{G}_{0};x\xleftarrow{\$} A( y)\right]
\end{gather*}
The assumption holds when ${\mathrm{Adv}^{dlog}_{\mathbb{G}_{0}}(\mathcal{A})}$ is negligible.
\end{definition}

\subsection{Multi-level CP-ABE}
The multi-layered and intertwined doctor-patient relationships across different healthcare providers make it impractical to protect EHRs. Each distinct part of the EHR file may require to be accessed with completely different access rights depending on the purpose of the data user. Therefore, the naive CP-ABE is not fully compatible in our scenario. However, as one of the successors to CP-ABE, ML-ABE fills in the gaps.
\begin{definition}[ML-ABE~\cite{Kaaniche2017}]
ML-ABE consists of four algorithms (\textbf{setup}, \textbf{encrypt}, \textbf{keygen}, \textbf{decrypt}):
\begin{itemize}
\item \textbf{setup}: This algorithm is executed by a trusted authority to generate public parameters $\mathfrak{pp}$ and a master key $\mathfrak{msk}$ according to the security parameter $\mathfrak{K}$.
\item \textbf{encrypt}: It is invoked by the data owner to encrypt the plaintext $\mathbb{M}=\{m_{l}\}_{l\in \{1,c\}}$ with respect to the multi-level security, where c represents the number of security levels. There are four required inputs, $\mathfrak{pp}$, the plaintext $\mathbb{M}$, the access tree $\mathbb{A}$ defined by the data owner over the universe of attributes $\mathbb{S}$, and the set of security levels $\{k_{l}\}_{l\in \{1,c\}}$. It returns the enciphered data $\mathbb{C}:=\{\mathbb{A} ,\forall k_{l} :\{\mathbb{A}^{\prime\prime}_{i}\}_{l} ,\mathbb{C}_{l}\}$. Here we underline that $\{\mathbb{A}^{\prime\prime}_{i}\}_l$ is a set of required sub-trees that must be satisfied by each security level $k_l$ for ${l\in \{1,c\}}$. We also provide the definition of their access structure below.
\item \textbf{keygen}: This algorithm is performed by the trusted authority to generate and issue the decryption key for the users depending on a set of attributes $\mathbb{S}$. It takes as input a set of attributes $\mathbb{S}$, the public parameters $\mathfrak{pp}$, and the master key $\mathfrak{msk}$ generated previously. The output will be the corresponding decryption key $\mathcal{DK}$ for a specific user or entity involved in the system.
\item \textbf{decrypt}: This algorithm is called by the data user to decrypt the ciphertext with respect to the multi-level security. There are three required inputs, the public parameters $\mathfrak{pp}$, the ciphertext $\mathbb{C}$, and the decryption key $\mathcal{DK}$. Note that $\mathbb{C}$ is packed with the relevant access policy $\mathbb{A}$, the security level $k_l$, and a set of required sub-trees $\{\mathbb{A}^{\prime\prime}_{i}\}_l$. It outputs the plaintext $m_l$ by decrypting the corresponding ciphertext $\mathbb{C}_l$ if the deciphering entity's attributes meet the requisites described in $\mathbb{C}$.
\end{itemize}
\end{definition}

\begin{definition}[Access Structure~\cite{Kaaniche2017}]
Let $\mathbb{A}$ be the access structure with multi-threshold security levels $k_l$, $l \in \{1,c\}$. Let $\mathbb{A}^{\prime}_{x}$ be the sub-tree of $\mathbb{A}$ rooted at a particular node x. Also, let $\{\{\mathbb{A}^{\prime\prime}_{i}\}_{l}\}$ be the sub-trees within the outer level. The root node is an AND gate defined as a \text{$k_l$-out-of-c} security levels. $p_l$ subsets of attributes and $n_l$ sub-trees of the root node are required to reconstruct the corresponding secret sharing embedded in the ciphertext $\mathbb{C}$ for security level $k_l$. $\mathbb{A}^{\prime}_{x}(\mathbb{S})=1$ if and only if a set of attributes $\mathbb{S}=\{a_{i}\}_{i\in \{1,l\}}$ satisfies the sub-tree and the number of attributes l is at least as many as the number of children of node x, otherwise $\mathbb{A}^{\prime}_{x}(\mathbb{S})=\bot$.
\end{definition}

\subsection{Multi-Signature}
A Multi-Signature (MS) solution allows a group of signers to co-sign on a common document in a compact manner~\cite{Bellare2006}. As a real-life example, the publication of a report/document often requires the cooperation of multiple colleagues. In order to guarantee the authenticity of the information in the report, each participant needs to sign the file. Therefore, Multi-Signature technology is used to fulfil this type of requirement in the electronic world. Besides, the ABE approach described in the previous section has already reduced the cost of key management by providing one-to-many encrypted access control~\cite{li2013fine}. Thus, we prefer to use a Multi-Signature scheme that is not based on comparatively more burdensome requirements of PKI (e.g., knowledge of secret key hypothesis~\cite{Boldyreva2002}) to enhance the practicality of the proposed E-Tenon system further. Bellare and Neven's MS-BN~\cite{Bellare2006} defined below fits well with our concept.

\begin{definition}[MS-BN~\cite{Bellare2006}]
MS-BN is a scheme consisting of four randomised algorithms (\textbf{Pg}, \textbf{Kg}, \textbf{Sign}, \textbf{Vf}):
\begin{itemize}
\item \textbf{Pg}: This algorithm is executed by a trusted authority to generate global parameters and output $\mathbb{G},p,g$, where $\mathbb{G}$ is a multiplicative cyclic group of prime order p, and g is a generator of $\mathbb{G}$ chosen at random.
\item \textbf{Kg}: This algorithm is called by each signer and co-signer to produce their own key pair used in the signing process. It outputs the signing key $\mathcal{SK}:=r\xleftarrow{\$}\mathbb{Z}_{p}$ randomly chosen from the finite field $\mathbb{Z}_p$ and the related verification key $\mathcal{VK}:=g^{\mathcal{SK}}$.
\item \textbf{Sign}: This algorithm is performed by the signers, and there are three rounds of communication. Each signer will perform some computation in the local scope based on messages shared by all co-signers as well as share their own message with others. It takes as input a signing key of the current signer $\mathcal{SK}_i$, a list of verification keys of all involved signers $\mathbb{V}:=\{\mathcal{VK}_1,\mathcal{VK}_2,\cdots,\mathcal{VK}_n\}$, and a message $\mathfrak{msg}$ to be multi-signed. It outputs the compact signature $\sigma$ consisting of the nonce commitments and the signatures if everyone is honest, otherwise it outputs $\bot$.
\item \textbf{Vf}: This algorithm is executed by the verifiers. There are three required inputs: a message $\mathfrak{msg}$, a compact signature to be verified $\sigma$, and a set of verification keys of all involved signers $\mathbb{V}:=\{\mathcal{VK}_1,\mathcal{VK}_2,\cdots,\mathcal{VK}_n\}$. It returns 1 to indicate the signature $\sigma$ is valid, otherwise it returns $\bot$.
\end{itemize}
\end{definition}
Here we stress two important facts about MS-BN. First, the security of this scheme is guaranteed on the assumption that at least one of the signers is honest~\cite[Sec.~4]{Bellare2006}. Second, the \textit{Kg} algorithm of MS-BN is run independently by each signer to generate the key pair. Such an assumption leads to a breach of security when all the signers are honest-but-curious or dishonest. In view of the increasing sophistication of cyber attacks, any end-user can no longer be undoubtedly trusted. Hence, our model will strengthen MS-BN to accommodate the case where no particular signer is fully trusted. To achieve that, we do not allow the non-trusted signer to perform the \textit{Kg} algorithm without the support of a trusted entity. In other words, the secret keys required for the user to operate the ABE and Multi-Signature related algorithms will be issued by an Attribute Authority (AA) at once where necessary.

\begin{figure*}[t]
    \centering
    \includegraphics[width=0.7\textwidth]{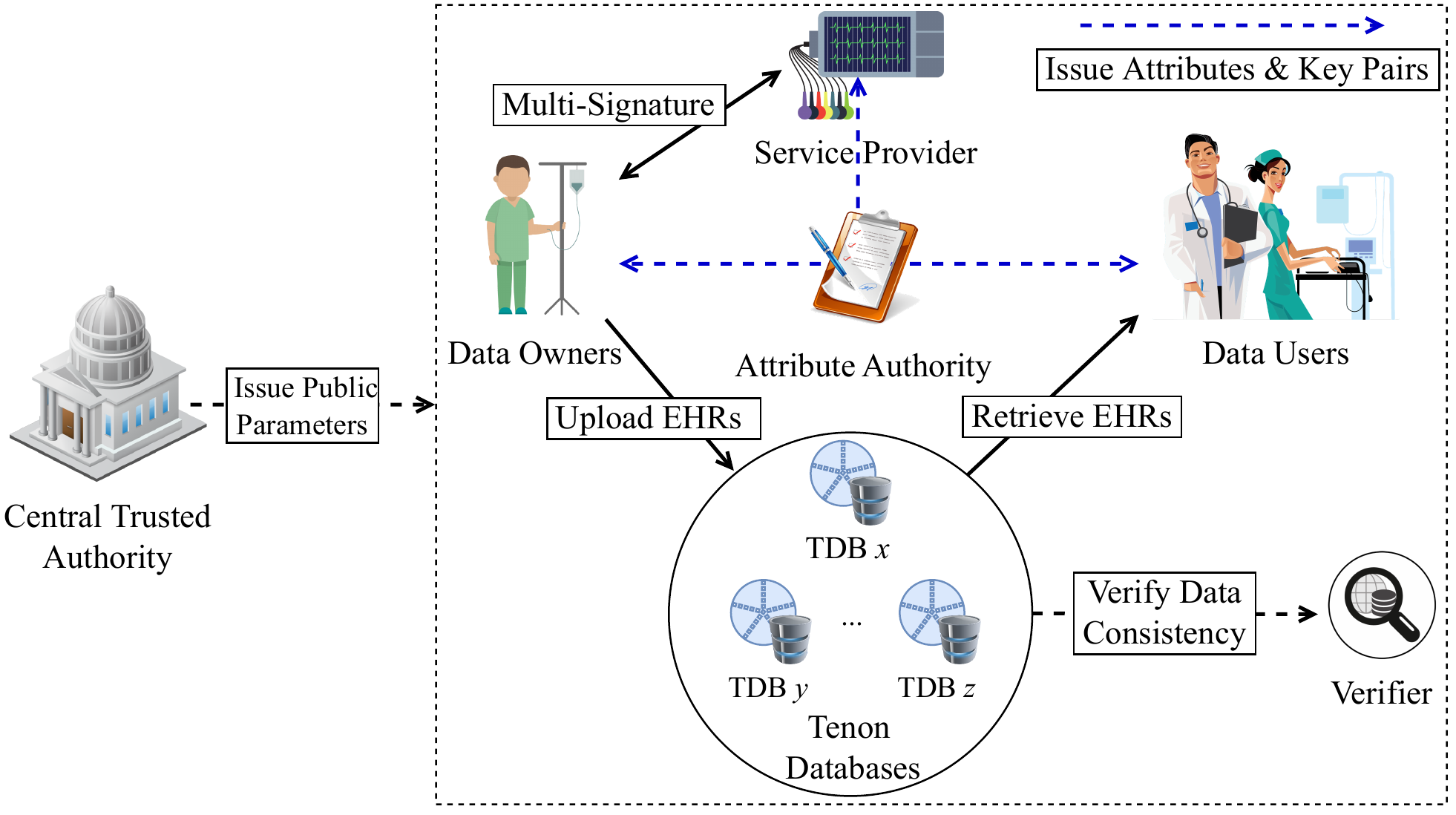}
    \caption{System model of the proposed scheme.}
    \label{sysmodel}
\end{figure*}

\section{System and Adversarial Model}\label{section:3}
In this section, we provide a high-level overview of the proposed system model with respect to entities involved in E-Tenon (as depicted in Fig. \hyperref[sysmodel]{2}). Afterwards, we analyse security considerations along with an adversarial model.

\subsection{System Model}
To establish the system model, we first introduce an efficient open database, then we merge and extend a Multi-Signature scheme MS-BN~\cite{Bellare2006} with an encryption scheme ML-ABE~\cite{Kaaniche2017}. Our system (as depicted in Fig. \hyperref[sysmodel]{2}) ends up with three distinct phases: SETUP, ACCUMULATION and RETRIEVAL, along with seven secure algorithms. In addition, there are six crucial entities: Central Trusted Authority (CTA), Attribute Authority (AA), Data Owner (DO), Service Provider (SP), Data User (DU), and Tenon Database (TDB). Besides, we allow for the option of a seventh participant: Verifier (VER).
\begin{itemize}
\item \textbf{CTA} is a fully trusted entity responsible for generating system-wide public parameters for all participants within the system.
\item \textbf{AA} serves in a similar way as the CTA. It is in charge of managing the user's attributes and issuing secret keys for the user, where appropriate. We note that the state-of-the-art multi-authority ABE systems use several different AAs and make each AA responsible for only one specific attribute. However, it must resist collusion attacks. In our case, we do not require multiple AAs, and we consider the AA as a trusted entity.
\item \textbf{DO} is the actual owner of the EHRs, i.e., the patient. Typically, DOs are concerned about the privacy of their EHRs, and they have the right to control the sharing of their EHRs. However, DOs can also be malicious. For example, DOs may upload incorrect EHRs to mislead data users into making improper treatment decisions. In E-Tenon, DOs can preprocess and selectively encrypt EHRs with self-defined multi-level access policies before sending the data to the database. DOs will also be required to multi-sign their data.
\item \textbf{SP} is an honest-but-curious entity involved in the signing process. It provides unconfirmed EHR to the DO. For example, a smart blood pressure sensor provides readings to the DO (such data remain subject to patient confirmation). However, one exceptional SP who can be trusted is the patient's doctor in charge (they provide patients with officially confirmed diagnostic results and treatment plans).
\item \textbf{TDB} is an honest-but-curious entity responsible for the data management. TDB per se is a distributed open database. The data should be stored as it is, and TDB has no right to decrypt any of the secret relationships. We are inspired by the ancient timber mortise and tenon joints, a strong and stable way of joining multiple elements together by using a proper combination of concave and convex pieces as shown in Fig. \hyperref[tenon]{3}, when designing the TDB and introducing the \emph{Electronic Tenon Structure} for different EHR blocks to be securely joined together. By secure, we mean that no public data can be exploited by unauthorised entities as only data users with the appropriate attributes know the proper way to assemble the relevant EHR blocks.
\begin{figure}[ht]
    \centering
    \includegraphics[width=0.35\textwidth]{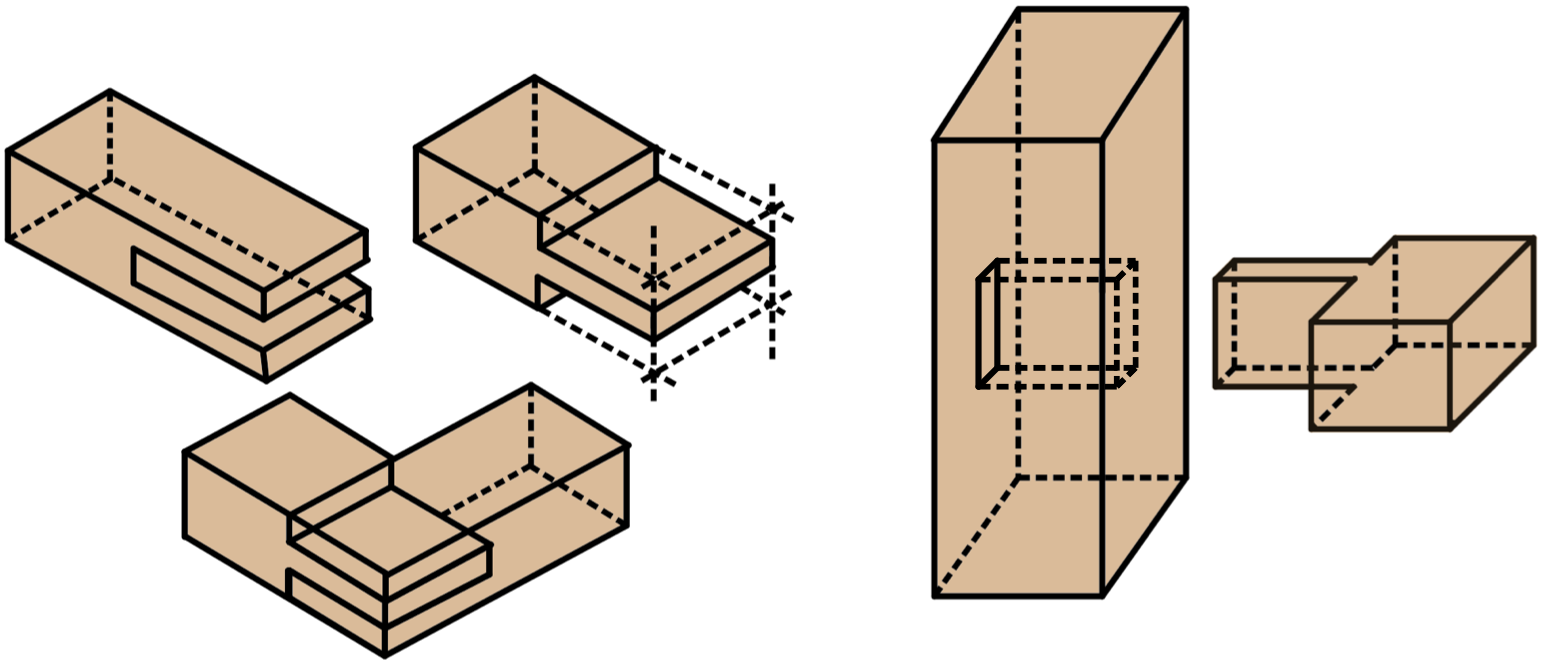}
    \caption{An example of mortise and tenon joints.}
    \label{tenon}
\end{figure}
\item \textbf{DU} is an individual or organisation (e.g., doctor, hospital, research institution, pharmaceutical and medical insurance company) that needs access to patient-owned EHRs in the TDB. DU requires an appropriate level of access, represented by their attributes, to reveal the secret relationships between EHR blocks. For example, a doctor may be able to extract five secret pointers to find and link five EHR blocks. However, a nurse may only be able to decrypt two pointers. Thus, there is a restriction on the amount of data that can be recovered due to the different attributes they hold. Moreover, a DU without the required attributes will be considered malicious when attempting to decrypt the pointers.
\item \textbf{VER} is a trusted participant who is responsible for auditing data consistency between multiple TDBs. The synchronisation of EHR across multiple databases enhances availability and avoids single points of failure.
\end{itemize}

\subsection{Adversarial Model}\label{Adversarial} E-Tenon is intended to be used by patients and a wide range of healthcare institutions. The novelty lies in the fact that most of the EHRs in the TDB are publicly accessible. Besides, we do not restrict EHRs to be transferred only within private networks such as the corporate Local Area Network. Accordingly, the vast majority of EHRs can be transmitted through untrusted public networks such as the Internet. While these considerations greatly increase the applicability and the efficiency of the model, they also expose system interactions and EHRs in transit to a variety of malicious cyber attackers. Therefore, our system must defend against the following threats:
\begin{itemize}
\item \textbf{Confidentiality Threat}: The system may fail to guarantee the secrecy of secret relationships between EHR blocks. For instance, a semi-trusted TDB may intend to discover as much information as possible while complying with the defined protocols. A malicious DU without appropriate permissions may attempt to exploit the open data and reveal the secret relationships.
\item \textbf{Privacy Threat}: DO and DU's identity may be revealed when they interact with a semi-trusted TDB. A malicious DU may be able to infer a relationship between the patient and the data stored in the TDB.
\item \textbf{Integrity and Authenticity Threat}: As EHR is patient-centric data, the patient has primary control over it. However, it remains a challenge to ensure the integrity and authenticity of the EHR provided by patients. One possible attack is that the EHR is tampered with by an intermediary when transmitted over insecure public channels. Even worse, patients themselves may deliberately alter their EHR before uploading in order to obtain \textit{biased} diagnosis and then obtain a large insurance claim (they may also deny that they have uploaded fake data). In this context, although we can use digital signatures to resist these attacks, they may be forged.
\end{itemize}

\subsection{Security Assumptions}
Some of the key assumptions are summarised as follows:
\begin{itemize}
\item DOs and DUs are expected to be educated about privacy rights and obligations. Thus, they will not actively disclose any confidential information to unaffiliated and unauthorised third parties.
\item DOs can apply appropriate access policies to different categories of EHRs according to a \textit{layman-friendly} guidebook provided by the administrator.
\item The semi-trusted TDB and unauthorised DUs cannot infer the type of EHR when each category of data contains at least $\kappa$ different types.
\item The diagnosis will only be provided to the patient after medical experts have confirmed it. Moreover, as a trusted SP, a doctor in charge in an ideal state will not be bribed by anyone to provide a fake diagnosis to the patient.
\end{itemize}

\subsection{Security Games}
Based on the system and adversarial models, we consider the following security games to define the security notion of our E-Tenon system.

1) To prove that E-Tenon is secure against confidentiality and privacy threats, we define a IND-CCA-1 security game between a challenger $\mathcal{C}$ and an adversary $\mathcal{A}$:
\begin{itemize}
    \item \textbf{Setup:} $\mathcal{C}$ runs setup algorithm, and sends the public parameters $\mathfrak{pp}$ to $\mathcal{A}$.
    \item \textbf{Query:} $\mathcal{C}$ initialises an empty table $T$, an integer session counter $j$ starting from zero and an empty set $\mathbb{Q}$. $\mathcal{A}$ can repeatedly query the following:
        \begin{itemize}
            \item \textbf{Create:} $\mathcal{C}$ increments $j$ by 1. $\mathcal{C}$ runs setup to obtain $\mathfrak{pp}$ and a master key $\mathfrak{msk}$, then it runs keyGeneration to extract a decryption key $\mathcal{DK}$ on $\mathbb{S}$ and the corresponding security levels $k_l$. $\mathcal{C}$ finally stores the entry $(j, \mathbb{S}, \mathfrak{pp}, \mathfrak{msk}, \mathcal{DK})$ in $T$ if it is not a duplicate entry.
            \item \textbf{Corrupt:} $\mathcal{A}$ requests the decryption output of a ciphertext $\mathbb{C}$ using $\mathcal{DK}$ on $\mathbb{S}$. $\mathcal{C}$ sets $\mathbb{Q} = \mathbb{Q}\cup\mathbb{S}$ if the $\mathcal{DK}$ for $\mathbb{S}$ exists in $T$ and proceeds.
            \item \textbf{Decrypt:} $\mathcal{C}$ decrypts $\mathbb{C}$ and outputs the results of the decryption to $\mathcal{A}$. Note this oracle can only be accessed before $\mathcal{A}$ receives the challenge ciphertext.
        \end{itemize}
    \item \textbf{Challenge:} $\mathcal{A}$ chooses two plaintext message $\mathbb{M}_0$ and $\mathbb{M}_1$ of the same length. $\mathcal{A}$ also submits a challenge access structure $\mathbb{A}^*$ such that $\mathbb{S}$ does not satisfy $\mathbb{A}^*$ for all $\mathbb{S}\in\mathbb{Q}$. $\mathcal{C}$ then randomly selects a bit $b\in\{0,1\}$ and outputs the encryption results of $\mathbb{M}_b$ under $\mathbb{A}^*$ and $k_l$ to $\mathcal{A}$.
    \item \textbf{Guess:} $\mathcal{A}$ outputs its guess $b^\prime\in\{0,1\}$ for $b$. $\mathcal{A}$ wins the game if $b^\prime=b$.
\end{itemize}
\begin{definition}\label{cca-1}
ML-ABE is CCA-1 secure against confidentiality and privacy threats, if for all PPT adversaries, there is a negligible function in winning the security game defined above, such that $$\mathrm{Adv}^{CCA-1}_{\mathcal{A}}(\lambda)=\mathrm{Pr}[b^\prime=b]=\frac{1}{2}\pm\epsilon$$
\end{definition}

2) To prove that E-Tenon is secure against integrity and authenticity threats, we define a MU-UF-CMA security game between a challenger $\mathcal{C}$ and a forger $\mathcal{F}$:
\begin{itemize}
    \item \textbf{Setup:} $\mathcal{C}$ runs setup and keyGeneration algorithms, and sends the public parameters $\mathfrak{pp}$, a random secret key $\mathcal{SK^*}$ and a public key $\mathcal{VK^*}$ to an honest signer. $\mathcal{VK^*}$ is also shared with $\mathcal{F}$.
    \item \textbf{Attack:} $\mathcal{F}$ initialises a message $\mathfrak{msg}$ to be multi-signed and a set containing the public keys of all co-signers $\mathbb{V} = \{\mathcal{VK}_1,\cdots,\mathcal{VK}_n\}$ where $\mathcal{VK^*} \in \mathbb{V}$. Note that all keys in $\mathbb{V}$ are controlled by $\mathcal{F}$ except for $\mathcal{VK^*}$. Meaning that $\mathcal{F}$ impersonates other co-signers with these keys to run the multiSign algorithm with the honest signer. It either outputs a signature $\sigma$ or a $\bot$. 
    \item \textbf{Forgery:} Once the above phase terminates, $\mathcal{F}$ outputs its forgery $(\mathbb{V}, \mathfrak{msg}, \sigma)$. $\mathcal{F}$ wins the game if the forgery passes the verify algorithm.
\end{itemize}
\begin{definition}\label{uf-cma}
MS-BN is MU-UF-CMA secure against integrity and authenticity threats, if for all PPT adversaries, there is a negligible function in winning the security game defined above, such that $$\mathrm{Adv}^{MU-UF-CMA}_{\mathcal{F}}(\lambda)=\mathrm{Pr}[\mathrm{verify}(\mathbb{V}, \mathfrak{msg}, \sigma) = 1] \leq \epsilon$$
\end{definition}

\section{Concrete Construction}\label{section:4}
Our system incorporates three important phases and seven secure algorithms. We describe the construction details of each phase separately with further specifications in the following subsections. Table \hyperref[natations]{1} lists some essential notations and cryptographic functions we used.

\subsection{Overview}
We propose an electronic tenon system (E-Tenon) that benefits from the effective integration of multi-level attribute-based encryption and multi-signature techniques. Our innovations and extensions to them allow these existing technologies to work appropriately with open databases.
\begin{table}[t]
    \centering
    \caption{Notations and cryptographic functions.\label{natations}}
    \label{tab:my_label}
    \begin{tabular}{|c|c|}
        \hline
        \textbf{Notation}                     & \textbf{Definition}                    \\
        \hline\hline
        $H(\cdot)$                            & One-way hash function                  \\
        \hline
        $\|$                                  & Concatenation operation                \\
        \hline
        $\bot$                                & Bottom constant of propositional logic \\
        \hline
        $\{k_{l}\}_{l\in \{1,c\}}$            & Set of security levels                 \\
        \hline
        $c$                                   & Number of security levels              \\
        \hline
        $\mathcal{SK}$                        & Signing key                            \\
        \hline
        $\mathcal{VK}$                        & Verification key                       \\
        \hline
        $\mathcal{EK}$                        & Encryption key                         \\
        \hline
        $\mathcal{DK}$                        & Decryption key                         \\
        \hline
        $\mathcal{SDK}$                       & Signing and decryption key pair        \\
        \hline
        $\mathcal{VEK}$                       & Verification and encryption key pair   \\
        \hline
        $\mathbb{A}$                          & Patient-defined access structure       \\
        \hline
        $\{\mathbb{A}^{\prime\prime}_{i}\}_l$ & Sub-trees, sub-access structure        \\
        \hline
        $\mathbb{S}$                          & Universe of attributes                 \\
        \hline
        $\mathbb{V}$                          & Verification key set                   \\
        \hline
        $\mathfrak{N}_i$                      & A unique pointer                       \\
        \hline
        $\Phi_i$                              & A tokenised EHR block                  \\
        \hline
        $\sigma$                              & Multi-Signature                        \\
        \hline
        $\gamma,\delta,\mathfrak{r}$          & Random exponents                       \\
        \hline
        $\epsilon$                            & A negligible number                    \\
        \hline
    \end{tabular}
\end{table}

To the best of our knowledge, the existing ABE-based privacy-preserving systems pose many redundant encryption and decryption overheads. However, our solution confidently allows EHRs to be securely made open (without encryption) in the TDB after special preprocessing. In concrete terms, EHR blocks stored in the TDB can only be mapped into meaningful information by deciphering relevant secret pointers. In addition, data shuffling techniques are applied to change the position and order of EHR blocks constantly. This signifies that the open data is presented to DUs at random each time the TDB is accessed. Furthermore, in the original MS-BN scheme, there must be a trusted signing entity involved in the signing process, but we cannot assume that this will always be the case in a safety-critical application. Therefore, our solution does not necessarily need the presence of a fully trusted signer to ensure the unforgeability of a multi-signature, which makes our E-Tenon system more flexible and practical. Eventually, we present steps grounded on sound logic to ensure that the SPs and DOs can always sign the same message. Taken together, these provide us with the ability to manage EHRs in an efficient, flexible and granular manner while maintaining privacy and security at the same time.

\subsection{Workflow of E-Tenon}
The workflow of our E-Tenon system is presented in Fig. \hyperref[workflow]{4} where green entities are fully trusted, red entities may be malicious, and blue entities are honest-but-curious. During the SETUP phase, the CTA and AA will generate and issue the public parameters, attributes and keys required by all system users. In the next stage, named ACCUMULATION, a total of four fundamental algorithms are used. Before the secret relationships between EHR blocks can be established, it first needs to be classified into two main categories: identifiable data and Non-PII data. The EHR preprocessing algorithm will recommend patients to perform minor encryption on identifiable data as well as the secret relationships between EHR blocks. The Non-PII data classified by our algorithm will be made open after preprocessing as it can not be used to trace a patient's identity. Note that when encryption is performed with a patient-defined access policy, it is equivalent to the patient giving consent to those users who satisfy the access policy. Subsequently, the DO and SP multi-sign the data such that the TDB can refuse to store the data if the signature is found to be invalid or forged. Apart from this, signers may also refuse to sign if they believe the data is illegally modified. At the final RETRIEVAL stage, the DUs also have the option to verify the signature of the data, and they can decrypt the pointers at different security levels according to their attributes when they believe that the signature is legitimate. Then, the decrypted pointers can be used to find and combine the relevant EHR blocks in the proper order to recover the correct information.
\begin{figure}[ht]
    \centering
    \includegraphics[width=\columnwidth]{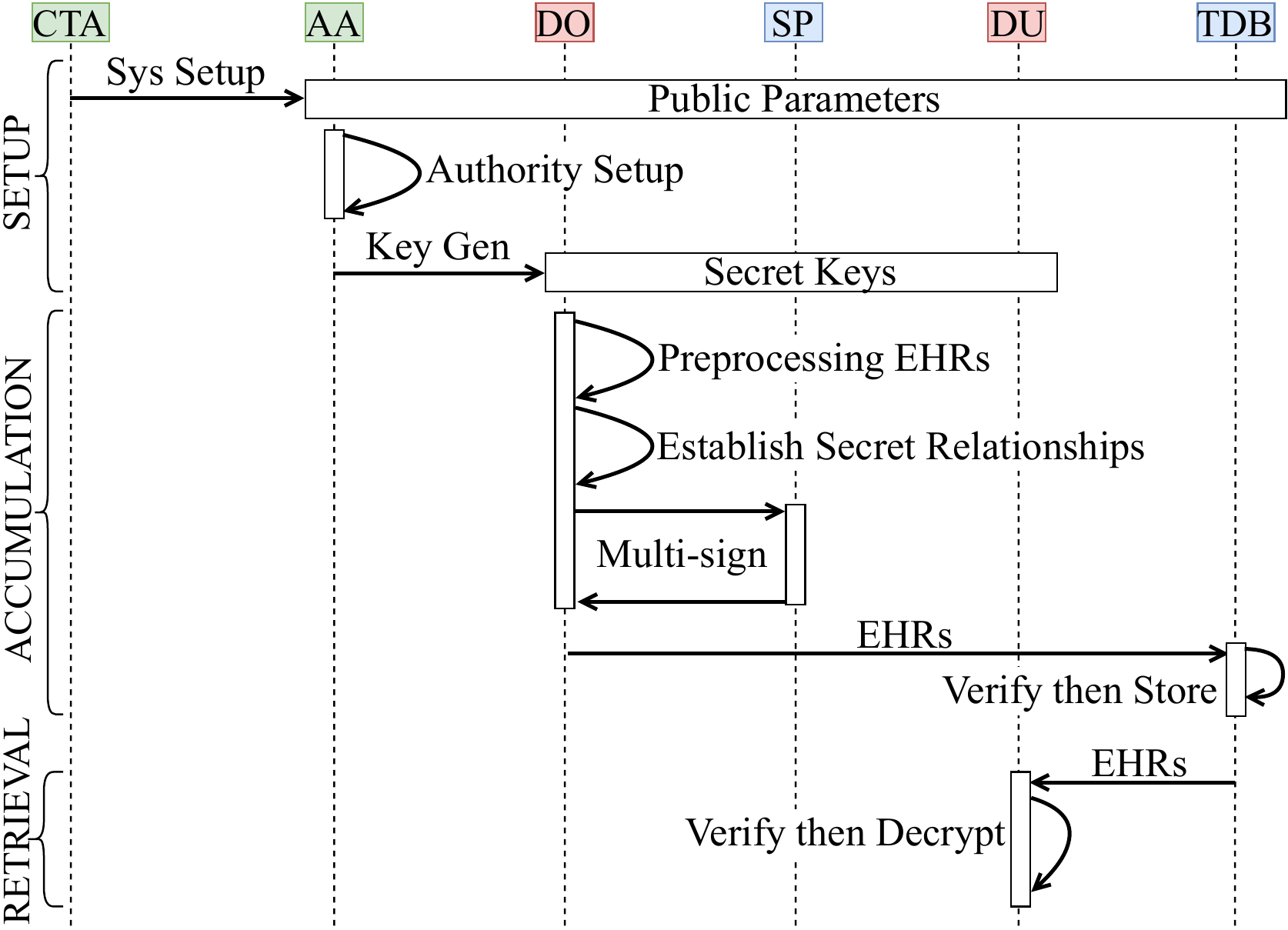}
    \caption{Workflow of the proposed scheme.}
    \label{workflow}
\end{figure}

\subsection{SETUP Phase}
Let $\lambda$ be the implicit security parameter that denotes the size of the cryptographic groups, and let $\mathbb{S}:=\{a_1,a_2,\cdots,a_n\}$ be the universe of the entity's attributes. The following two algorithms need to be administered by the CTA and AA for the initial system and authority setup process of the proposed scheme.

\begin{itemize}
    \item \textbf{setup}($\lambda$): It initially selects a generator $g${\scriptsize$\ \xleftarrow{\$}\ $}$\mathbb{G}_{0}^{*}$ and two unique elements $\gamma$ and $\delta$, at random $\gamma,\delta${\scriptsize$\ \xleftarrow{\$}\ $}$\mathbb{Z}_{p}$. Then, the master key $\mathfrak{msk}$ is defined as $\mathfrak{msk}:=(\delta,g^{\gamma})$. Finally, the public parameters $\mathfrak{pp}$ are grouped into the following seven auxiliary elements $\mathfrak{pp}:=\{\mathbb{G}_0,\mathbb{G}_1,p,g,g^{\delta},e,e(g,g)^{\gamma}\}$. $\mathfrak{pp}$ then is made public at system level and $\mathfrak{msk}$ can be used to create decryption keys according to user attributes.

    \item
          \textbf{keyGeneration}($\mathfrak{pp}$,$\mathfrak{msk}$,$\mathbb{S}$): This algorithm can be executed by either the AA or the signing parties depending on whether a trusted SP is involved in the signing process or not. In the first case, AA uses this algorithm to produce two \textit{distinct} pairs of keys (i.e., $\mathcal{SDK}$, the signing and decryption key pair, and $\mathcal{VEK}$, the verification and encryption key pair) once $\mathfrak{pp}$ and $\mathfrak{msk}$ are successfully generated by the CTA. It starts by choosing one random $\mathfrak{r}$ and a set of randoms $\{\mathfrak{r}_{a}\}$ from the finite field $\mathbb{Z}_p$ where each $a$ is in $\mathbb{S}$ such that $\forall a\in \mathbb{S}:\mathfrak{r} ,\mathfrak{r}_{a}${\scriptsize$\ \xleftarrow{\$}\ $}${\mathbb{Z}_{p}}$. These are used to randomise private keys and prevent DOs from compromising data confidentiality by colluding.
          All necessary keys for the user to operate both ABE and Multi-sig algorithms are formed along the following lines:
          {\small \begin{gather*}
              keys:=\{\mathcal{SDK} =(\mathcal{SK} ,\mathcal{DK}) ,\mathcal{VEK} =(\mathcal{VK} ,\mathcal{EK})\}\\
              \begin{cases}
                  \mathcal{SK} =\mathfrak{r},  \ \mathcal{VK} =g^{\mathcal{SK}},  \ \mathcal{EK} =\mathfrak{pp}                                                                                                                            & \\
                  \mathcal{DK} =\left\{\mathcal{D} =g^{\frac{\gamma +\mathfrak{r}}{\delta }} ,\forall a\in \mathbb{S} :\mathcal{D}_{a} =g^{\mathfrak{r}} .H(a)^{\mathfrak{r}_{a}} ,\mathcal{D}_{a}^{\prime } =g^{\mathfrak{r}_{a}}\right\} &
              \end{cases}
          \end{gather*}}

          \noindent where $\mathcal{VK}$ and $\mathcal{EK}$ can be made public, but $\mathcal{SK}$ and $\mathcal{DK}$ need to be kept secret. In the second case, if a trusted SP is involved in the multi-signature process, the signer may choose to generate his/her own signing key pair without relying on the AA. Despite that, $\mathcal{EK}$ and $\mathcal{DK}$ are still required to be issued by an AA.

\end{itemize}

\begin{figure*}[t]
    \centering
    \tikzset{every picture/.style={line width=0.75pt}} 
    \scalebox{1}{
        \begin{tikzpicture}[x=0.75pt,y=0.75pt,yscale=-1,xscale=1]

            \draw    (50,10) -- (50,285) ;
            \draw    (450,10) -- (450,285) ;
            \draw  [fill={rgb, 255:red, 255; green, 255; blue, 255 }  ,fill opacity=1 ] (22,15) -- (192,15) -- (192,55) -- (22,55) -- cycle ;
            \draw    (192,40) -- (450,40) ;
            \draw [shift={(450,40)}, rotate = 180] [color={rgb, 255:red, 0; green, 0; blue, 0 }  ][line width=0.75]    (10.93,-3.29) .. controls (6.95,-1.4) and (3.31,-0.3) .. (0,0) .. controls (3.31,0.3) and (6.95,1.4) .. (10.93,3.29)   ;
            \draw  [fill={rgb, 255:red, 255; green, 255; blue, 255 }  ,fill opacity=1 ] (300,51) -- (470,51) -- (470,91) -- (300,91) -- cycle ;
            \draw    (300,73) -- (53,73) [color={rgb, 255:red, 0; green, 0; blue, 0 }  ];
            \draw [shift={(51,73)}, rotate = 360] [color={rgb, 255:red, 0; green, 0; blue, 0 }  ][line width=0.75]    (10.93,-3.29) .. controls (6.95,-1.4) and (3.31,-0.3) .. (0,0) .. controls (3.31,0.3) and (6.95,1.4) .. (10.93,3.29)   ;
            \draw    (50,100) -- (450,100) ;
            \draw [shift={(450,100)}, rotate = 180] [color={rgb, 255:red, 0; green, 0; blue, 0 }  ][line width=0.75]    (10.93,-3.29) .. controls (6.95,-1.4) and (3.31,-0.3) .. (0,0) .. controls (3.31,0.3) and (6.95,1.4) .. (10.93,3.29)   ;
            \draw    (450,117) -- (53,117) [color={rgb, 255:red, 0; green, 0; blue, 0 }  ];
            \draw [shift={(51,117)}, rotate = 360] [color={rgb, 255:red, 0; green, 0; blue, 0 }  ][line width=0.75]    (10.93,-3.29) .. controls (6.95,-1.4) and (3.31,-0.3) .. (0,0) .. controls (3.31,0.3) and (6.95,1.4) .. (10.93,3.29)   ;
            \draw  [fill={rgb, 255:red, 255; green, 255; blue, 255 }  ,fill opacity=1 ] (22,125) -- (182,125) -- (182,190) -- (22,190) -- cycle ;
            \draw  [fill={rgb, 255:red, 255; green, 255; blue, 255 }  ,fill opacity=1 ] (300,180) -- (460,180) -- (460,245) -- (300,245) -- cycle ;
            \draw    (182,162) -- (450,162) ;
            \draw [shift={(450,162)}, rotate = 180] [color={rgb, 255:red, 0; green, 0; blue, 0 }  ][line width=0.75]    (10.93,-3.29) .. controls (6.95,-1.4) and (3.31,-0.3) .. (0,0) .. controls (3.31,0.3) and (6.95,1.4) .. (10.93,3.29)   ;
            \draw    (300,218) -- (52.67,218) [color={rgb, 255:red, 0; green, 0; blue, 0 }  ];
            \draw [shift={(50.67,218)}, rotate = 360] [color={rgb, 255:red, 0; green, 0; blue, 0 }  ][line width=0.75]    (10.93,-3.29) .. controls (6.95,-1.4) and (3.31,-0.3) .. (0,0) .. controls (3.31,0.3) and (6.95,1.4) .. (10.93,3.29)   ;
            \draw  [fill={rgb, 255:red, 255; green, 255; blue, 255 }  ,fill opacity=1 ] (22,250) -- (188,250) -- (188,280) -- (22,280) -- cycle ;
            \draw  [fill={rgb, 255:red, 255; green, 255; blue, 255 }  ,fill opacity=1 ] (300,250) -- (465,250) -- (465,280) -- (300,280) -- cycle ;

            \draw (39.2,0) node [anchor=north west][inner sep=0.75pt]  [xscale=0.8,yscale=0.8] [align=left] {$\textbf{DO}$};
            \draw (442,0) node [anchor=north west][inner sep=0.75pt]  [xscale=0.8,yscale=0.8] [align=left] {$\textbf{SP}$};
            \draw (110,40) node  [xscale=0.8,yscale=0.8] [align=left] {\begin{minipage}[lt]{153.23pt}\setlength\topsep{0pt}
                    {\small \textbf{Select}: $\displaystyle {\textstyle r_{i}\xleftarrow{\$}\mathbb{Z}_{p}}$}\\{\small \textbf{Compute}: $\displaystyle \mathcal{RC}_{i} =g^{r_{i}}$, $\displaystyle t_{i} =H_{0}(\mathcal{RC}_{i})$}\\{\small \textbf{Send}: $\displaystyle t_{i}$}\\
                \end{minipage}};
            \draw (215,25) node [anchor=north west][inner sep=0.75pt]  [font=\normalsize,xscale=0.8,yscale=0.8] [align=left] {Send $\displaystyle t_{i}$ to SP };
            \draw (385,75) node  [xscale=0.8,yscale=0.8] [align=left] {\begin{minipage}[lt]{150.73pt}\setlength\topsep{0pt}
                    {\small \textbf{Select}: $\displaystyle {\textstyle r_{j}\xleftarrow{\$}\mathbb{Z}_{p}}$}\\{\small \textbf{Compute}: $\displaystyle \mathcal{RC}_{j} =g^{r_{j}}$, $\displaystyle t_{j} =H_{0}(\mathcal{RC}_{j})$}\\{\small \textbf{Send}: $\displaystyle t_{j}$}\\
                \end{minipage}};
            \draw (215,58) node [anchor=north west][inner sep=0.75pt]  [font=\normalsize,xscale=0.8,yscale=0.8] [align=left] {\textcolor{blue}{Send $\displaystyle t_{j}$ to DO} };
            \draw (210,86) node [anchor=north west][inner sep=0.75pt]  [font=\normalsize,xscale=0.8,yscale=0.8] [align=left] {Send $\displaystyle \mathcal{RC}_{i}$ to SP };
            \draw (210,104) node [anchor=north west][inner sep=0.75pt]  [font=\normalsize,xscale=0.8,yscale=0.8] [align=left] {\textcolor{blue}{Send $\displaystyle \mathcal{RC}_{j}$ to DO} };
            \draw (389,216) node  [xscale=0.8,yscale=0.8] [align=left] {\begin{minipage}[lt]{159.8pt}\setlength\topsep{0pt}
                    {\small \textbf{Check}: $\displaystyle t_{i}\stackrel{?}{=} H_{0}(\mathcal{RC}_{i})$}\\{\small \textbf{Compute}: $\displaystyle \mathcal{RC} ={\textstyle \prod _{i=1}^{n}}\mathcal{RC}_{i}$}\\{\small $\displaystyle ch_{j} =H_{1}( \langle \mathbb{V} \rangle \parallel \mathcal{VK}_{SP} \parallel \mathcal{RC} \parallel \mathfrak{msg})$}\\{\small $\displaystyle \mathcal{MS}_{j} =\mathcal{SK}_{SP} \cdotp ch_{j} +r_{j}\bmod p$}\\{\small \textbf{Send}: $\displaystyle \mathcal{MS}_{j}$}\\
                \end{minipage}};
            \draw (210,148) node [anchor=north west][inner sep=0.75pt]  [font=\normalsize,xscale=0.8,yscale=0.8] [align=left] {Send $\displaystyle \mathcal{MS}_{i}$ to SP };
            \draw (210,203) node [anchor=north west][inner sep=0.75pt]  [font=\normalsize,xscale=0.8,yscale=0.8] [align=left] {\textcolor{blue}{Send $\displaystyle \mathcal{MS}_{j}$ to DO }};
            \draw (112,270) node  [xscale=0.8,yscale=0.8] [align=left] {\begin{minipage}[lt]{162.52pt}\setlength\topsep{0pt}
                    {\small \textbf{Compute}: $\displaystyle \mathcal{MS} ={\textstyle \sum _{i=1}^{n}}\mathcal{MS}_{i}\bmod p$ }\\{\small \textbf{Output}: $\displaystyle \sigma =(\mathcal{RC} ,\mathcal{MS})$}\\
                \end{minipage}};
            \draw (389,270) node  [xscale=0.8,yscale=0.8] [align=left] {\begin{minipage}[lt]{159.8pt}\setlength\topsep{0pt}
                    {\small \textbf{Compute}: $\displaystyle \mathcal{MS} ={\textstyle \sum _{i=1}^{n}}\mathcal{MS}_{i}\bmod p$ \ }\\{\small \textbf{Output}: $\displaystyle \sigma =(\mathcal{RC} ,\mathcal{MS})$}\\
                \end{minipage}};
            \draw (112,162) node  [xscale=0.8,yscale=0.8] [align=left] {\begin{minipage}[lt]{159.8pt}\setlength\topsep{0pt}
                    {\small \textbf{Check}: $\displaystyle t_{j}\stackrel{?}{=} H_{0}(\mathcal{RC}_{j})$}\\{\small \textbf{Compute}: $\displaystyle \mathcal{RC} ={\textstyle \prod _{i=1}^{n}}\mathcal{RC}_{i}$}\\{\small $\displaystyle ch_{i} =H_{1}( \langle \mathbb{V} \rangle \parallel \mathcal{VK}_{DO} \parallel \mathcal{RC} \parallel \mathfrak{msg})$}\\{\small $\displaystyle \mathcal{MS}_{i} =\mathcal{SK}_{DO} \cdotp ch_{i} +r_{i}\bmod p$}\\{\small \textbf{Send}: $\displaystyle \mathcal{MS}_{i}$}\\
                \end{minipage}};
        \end{tikzpicture}}
    \caption{Rounds of communication in multiSign algorithm (DO stands for Data Owner and SP stands for Service Provider).}
    \label{fig:multisig}
\end{figure*}
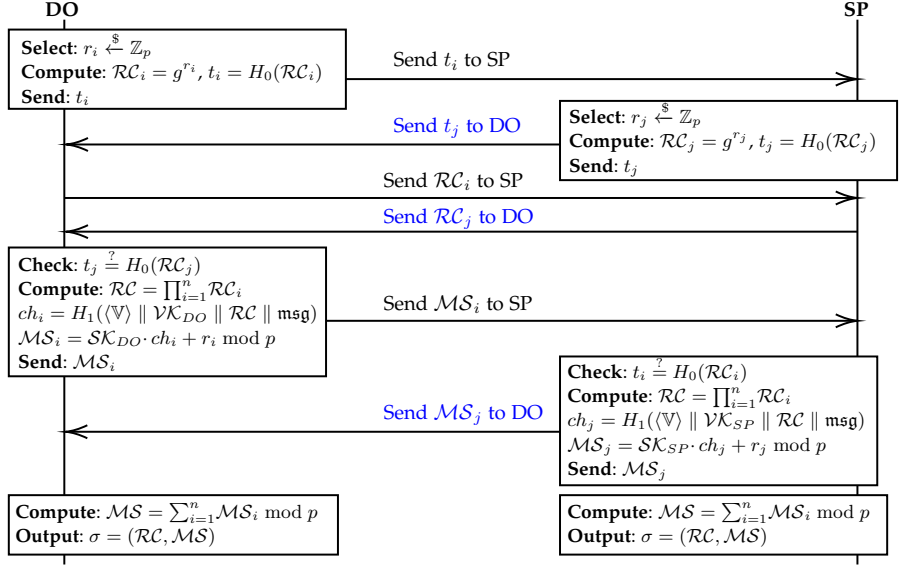
\subsection{ACCUMULATION Phase}
In order to understand what must be encrypted and what can be left open, we need to consider the ways in which data may be combined. For instance, an insecure combination is the National Insurance Number (NINO) with the medical condition since it reveals the patient's identity. However, blood pressure and symptoms can be seen as a safe combination. But it is noted that although the knowledge of a single symptom is not helpful in revealing a patient's identity (e.g., almost everyone may have a cough), detailed symptom information can be useful in inferring a patient's identity (e.g., it may be rare for a person to have a nosebleed, cough, fever and heart pain at the same time).

\begin{algorithm}
    \caption{dataPreprocessing($\Phi$)}\label{alg:cap}
    \begin{algorithmic}[1]
        \State \small \textbf{Begin:}
        \State \small \textbf{Step 1} classifies $\Phi$ by \textit{identifiable} and \textit{Non-PII}.
        \State \small \textbf{Step 2} If $tokenisable(\Phi) == true$, splits Non-PII $\Phi$ into data blocks, each block contains one main word plus the preceding stopwords if present.
        \State \small \textbf{Step 3} establishes the \textit{relationships} between an identifiable column and Non-PII columns as well as the \textit{relationships} between blocks of Non-PII columns.
        \State \small \textbf{Step 4} returns preprocessed/structured EHRs $\mathbb{M}$.
        \State \small \textbf{End.}
    \end{algorithmic}
\end{algorithm}

\begin{itemize}
    \item
          \textbf{dataPreprocessing}($\Phi$): This algorithm is run by the DO. It begins by classifying and labelling EHRs by identifiable and non personally identifiable information. As an example, identifiable columns may include the patient's NINO and mobile number. Non-identifiable columns include medical condition, gender, symptom and blood pressure. Next, it splits any tokenisable and Non-PII EHRs into blocks with the relationships between blocks linked by a 128-bit pointer (UUID). Instead of using pure numeric IDs that are easily guessed, we generate the Universally Unique Identifier using a cryptographically strong pseudo-random number generator provided in the Apache Commons IO library~\cite{uuid}. An example of a 128-bit UUID is 9458fdcc-6bed-46ec-b883-0076409e76f. This prevents simple brute-force guessing of the secret relationships because it is impossible to iterate through all random UUIDs. In the end, the preprocessed EHR blocks are output in a random access data structure, referred as \emph{electronic tenon structure}: $\mathbb{M} :=\{(\mathfrak{N}_1 ,\Phi _{1} ,\mathfrak{N}_x) ,(\mathfrak{N}_2 ,\Phi _{2} ,\mathfrak{N}_y) ,\cdots,(\mathfrak{N}_n ,\Phi _{n} ,\mathfrak{N}_z)\}$. Note that each element in $\mathbb{M}$ is a three-tuple containing (1) the UUID of the current EHR block, (2) the EHR block itself and (3) a pointer to the next EHR block.

    \item
          \textbf{encryptPointer}($\mathfrak{pp}$,$\mathbb{M}$,$\mathbb{A}$,$\{k_{l}\}_{l\in \{1,c\}}$): This algorithm is executed by the DO. It extracts the pointers $\{\mathfrak{N}_i,\mathfrak{N}_j,\mathfrak{N}_k,\cdots\}$ in $\mathbb{M}$, and encrypts them according to a patient-defined access structure $\mathbb{A}$ with different security levels $\{k_{l}\}_{l\in \{1,c\}}$, where pointers associated with different security levels require different attributes to decrypt. The ciphertext structure introduced in~\cite{Kaaniche2017} is adapted as below:
          {\small \begin{gather*}
              \mathbb{C}:=\left\{\mathbb{A}, \forall k_{l}:\{\mathbb{A}^{\prime\prime}_{i}\}_l, \mathbb{C}_{k_{l}}, \tilde{\mathbb{C}}_{k_{l}},\forall y: \mathbb{C}_{y}, \mathbb{C}_{y}^{\prime}\right\}\\
              \begin{cases}
                  \mathbb{C}_{k_{l}}=g^{\delta \varsigma_{l}}, \ \ \tilde{\mathbb{C}}_{k_{l}}=\mathfrak{N}_i \cdot e(g, g)^{\mathrm{\gamma\varsigma}_{l}} & \\
                  \mathbb{C}_{y}=g^{q_{y}(0)}, \ \ \mathbb{C}_{y}^{\prime}=H(\operatorname{att}(y))^{q_{y}(0)}                                            & \\
              \end{cases}
          \end{gather*}}
          where $g^{\delta}$ and $e(g,g)^{\gamma}$ are extracted from the public parameters $\mathfrak{pp}$ generated during the SETUP phase by the CTA. Moreover, we note that the advantage of CP-ABE is that the enciphering secret is built into the relevant ciphertext, rather than being placed in the private key (key management is minimised)~\cite{Bethencourt2007}. Here, the enciphering secret $\varsigma_{l}$ embedded in each ciphertext with a particular security level $k_l$ is computed as $\varsigma_{l}:=\sum_{i \in\left\{1,2, \cdots, n_{l}\right\}} q_{r}\left(\operatorname{index}\left(x_{i}\right)\right)$ where $q_{r}(x)$ is the polynomial related to the root node $r$ of $\mathbb{A}$, $q_{r}(x)=a_{0}+a_{1} x+\cdots+a_{d_{r}} x^{d_{r}}$~\cite{Kaaniche2017}.

    \item
          \textbf{multiSign}($\mathcal{SK}_i$,$\mathbb{V}$,$\mathfrak{msg}$): This algorithm requires several rounds of communication between sigining parties (e.g., DO and SPs). A compact multi-signature $\sigma$ is generated if all participants are honest, which means that the multiSign algorithm terminates immediately whenever one signer is dishonest. It takes as inputs a message $\mathfrak{msg}$, the current signer's signing key $\mathcal{SK}_i$, and a set of verification keys $\mathbb{V}:=\{\mathcal{VK}_1,\mathcal{VK}_2,\cdots,\mathcal{VK}_n\}$ of all participants. The multi-signature $\sigma:=(\mathcal{RC} \leftarrow \prod_{i=1}^{n} \mathcal{RC}_{i}, \mathcal{MS} \leftarrow \sum_{i=1}^{n} \mathcal{MS}_{i} \bmod p)$ is produced as a two-tuple containing the aggregated partial signatures $\mathcal{MS}$ and the nonce commitment $\mathcal{RC}$. It is generated based on the signing algorithm presented in Bellare and Neven's Multisig scheme~\cite{Bellare2006}, and the adapted version is shown in Fig. \hyperref[fig:multisig]{5}.
          In our system, there are two forms of data that need to be multi-signed: the EHR blocks per se and the ciphertext containing the secret relationships between them. Hence, we define $\sigma _{\Phi _{i}}$ as $\sigma _{\Phi _{i}}\leftarrow multiSign(\mathcal{SK}_{i} ,\mathbb{V} ,\mathfrak{msg}=H( \Phi _{i} \| \mathfrak{N}_{i} \| \mathfrak{pp} \| t))$ to represent the multi-signature for a given EHR block, and we define $\sigma _{E_{i}}$ as $\sigma _{E_{i}}\leftarrow multiSign(\mathcal{SK}_{i} ,\mathbb{V} ,\mathfrak{msg}=H( E_{i} \| \mathfrak{pp} \| t))$ to represent the multi-signature of the secret relationships. These ensure that DOs and SPs cannot refute their responsibility for the EHRs provided and allow TDB and DUs to verify the integrity and authenticity of the EHRs when necessary.

    \item
          \textbf{verify}($\sigma$,$\mathbb{V}$,$\mathfrak{msg}$): This deterministic algorithm is the last key algorithm in the ACCUMULATION phase. It can be executed by the TDB and DUs to verify the multi-signature $\sigma$. It starts by gathering the challenge numbers: $ch_{i} \leftarrow H_{1}\left(\langle\mathbb{V}\rangle \| \mathcal{VK}_{i}\| \mathcal{RC} \| \mathfrak{msg}\right)$ for $\forall i\in \{1,2,\cdots,n\}$ as in the third round of the signing process via an ideal cryptographic hash function $H_1:\{0,1\}^* \rightarrow \{0,1\}^{m\in\mathbb{N}}$. These challenge numbers are then applied to the final validation expression: $g^\mathcal{MS} \stackrel{?}{=} \mathcal{RC} \ {\prod_{i=1}^{n}} \mathcal{VK}_{i}^{ch_{i}}$. According to MS-BN, the verification fails ($\bot \leftarrow$ verify$(\sigma,\mathbb{V},\mathfrak{msg})$) if the above equation does not hold. The whole ACCUMULATION phase will also fail, and the data cannot be stored at this point. Therefore, legitimate EHRs can only be saved to the TDB if all the accompanying signatures $\sigma$ are validated by the TDB.

\end{itemize}

\subsection{RETRIEVAL Phase}
Once the DU confirms that the accompanying multi-signature is not a forgery, he/she can call the following algorithm to decrypt the ciphertext hierarchically. Please note that the higher the access rights represented by the DU's attributes, the larger the number of pointers that can be revealed.

\begin{itemize}
    \item
          \textbf{decryptPointer}($\mathfrak{pp}$,$\mathbb{C}$,$\mathcal{DK}_i$): It takes as input the public parameters $\mathfrak{pp}$, ciphertext $\mathbb{C}$, and the current decrypting entity's decryption key $\mathcal{DK}_i$. The inner ciphertext $\mathbb{C}_l$ can be decrypted, and the secret pointer $\mathfrak{N}_l$ can be retrieved if the DU's attributes embedded in $\mathcal{DK}_i$ satisfy the patient-defined access structure $\mathbb{A}$, with respect to the connected sub-trees $\{\mathbb{A}^{\prime\prime}_{i}\}_l$ and the security level $k_l$. More concretely, each security level needs to be evaluated separately for obtaining different $\mathfrak{N}_i$. Following the authors of ML-ABE, their algorithm starts decrypting from the outer level using the decryption algorithm developed in the classic CP-ABE proposed by Bethencourt, Sahai and Waters. For the internal level, the DU would be able to extract the enciphering secret $e(g, g)^{\mathfrak{r} \varsigma_{l}}$ from $n_l$ identified sub-trees $\{\mathbb{A}^{\prime\prime}_{i}\}_l$ rooted at the root node if the \text{$k_l$-security} level is satisfied~\cite{Kaaniche2017}:
          \begingroup
          \allowdisplaybreaks
          {\begin{align*}
                  F_{R_{k_{l}}} & = \prod_{x \in \{\mathbb{A}^{\prime\prime}_{i}\}_l} e(g, g)^{\mathfrak{r} q_{parent(x)}(index(x))}                                         \\
                                & = e(g, g)^{\Sigma_{x \in \{\mathbb{A}^{\prime\prime}_{i}\}_l} \mathfrak{r} q_{parent(x)}(index(x))} = e(g, g)^{\mathfrak{r} \varsigma_{l}}
              \end{align*} }
          \endgroup
          , where the function \textit{parent(x)} is called to find the parent node of node \textit{x} in $\mathbb{A}$. The index related to node $x$ is located by calling the function \textit{index(x)}. The secret $e(g, g)^{\mathfrak{r} \varsigma_{l}}$ can be used to derive a pointer $\mathfrak{N}_i$ that has been flagged with the specified security level.
          Having the secret key of the corresponding pointer extracted by a legitimate DU through the above steps, the pointer $\mathfrak{N}_i$ used to locate the corresponding EHR block can be obtained in its plaintext form by:
\end{itemize}
\begingroup
\allowdisplaybreaks
{\begin{align*}
        \frac{\tilde{\mathbb{C}}_{k_{l}}}{e\left(\mathbb{C}_{k_{l}}, \mathcal{D}\right) / F_{R_{k_{l}}}}
         & =\frac{\mathfrak{N}_i \cdot e(g, g)^{\mathrm{\gamma\varsigma}_{l}}}{e\left(g^{\delta \varsigma_{l}}, g^{(\gamma+\mathfrak{r}) / \delta}\right) / e(g, g)^{\mathfrak{r} \varsigma_{l}}} \\
         & =\frac{\mathfrak{N}_i \cdot e(g, g)^{\gamma\varsigma_{l}}}{e(g, g)^{\gamma\varsigma_{l}}}=\mathfrak{N}_i
    \end{align*}}
\endgroup

\begin{figure*}[ht]
    \centering
    \includegraphics[width=0.9\textwidth]{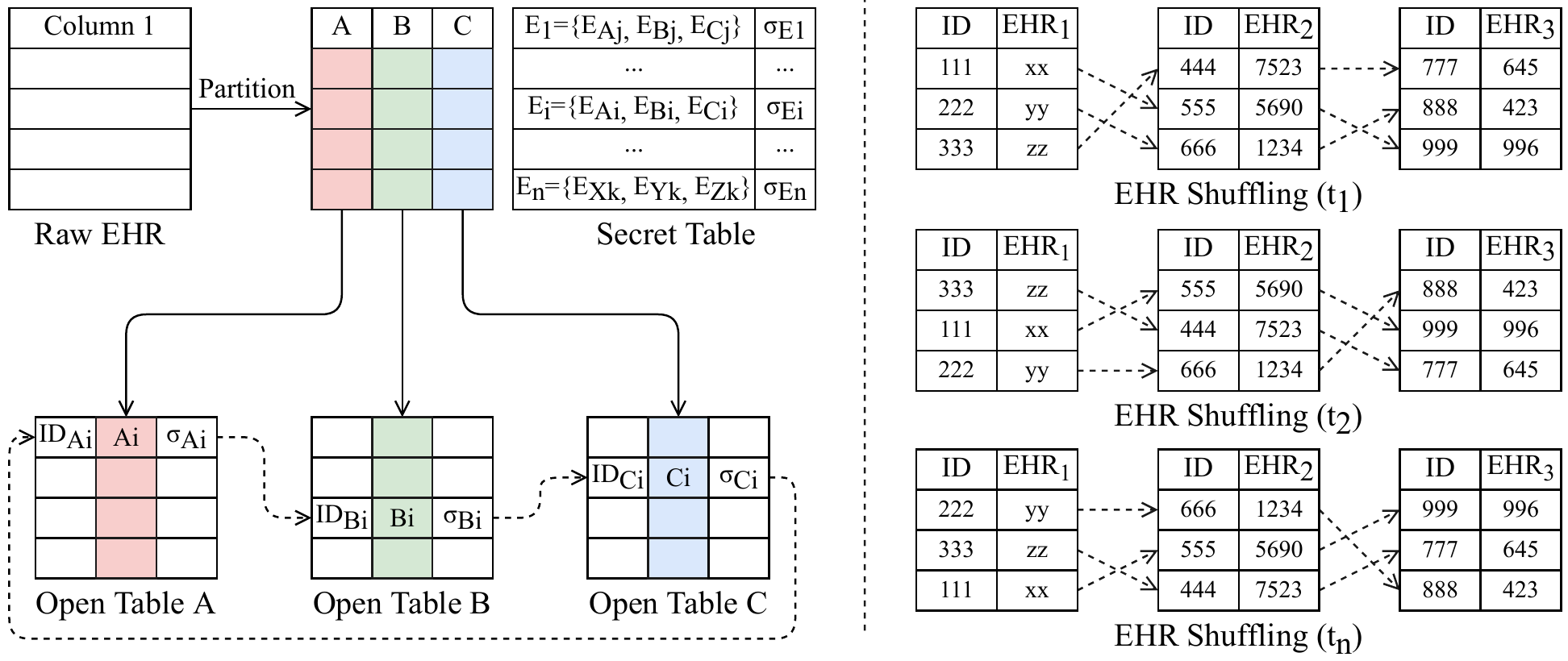}
    \caption{Example of working principle of the Tenon database.}\label{fig:tdb}
\end{figure*}
\subsection{Working Principle of TDB}\label{shuffling}
In this subsection, we explain the working principle of the TDB that forms one of the key components in the proposed E-Tenon system. As seen visually in the left part of Fig. \hyperref[fig:tdb]{6}, the TDB is composed of several open tables and one secret table. There are three columns per row in the open table: pointer, EHR block and multi-signature. It is worth noting that all encrypted data are separated from the open table. This is because we have adopted a multi-level ABE that produces a ciphertext containing multiple encrypted pointers. To reconstruct the data in the open tables, the authorised DU first decrypts the outer layer of the ciphertext. If successful, they will be presented with a series of encrypted pointers, and the number of pointers that can be decrypted depends on the DU's attributes. In this context, each row in the open table should not contain any encrypted pointers because this compromises the data confidentiality once a low privileged DU decrypts the outer ciphertext. Namely, an adversary can effortlessly use the encrypted pointers to locate the rows containing these pointers in the misconfigured open tables and directly combine them without the need to decrypt the secret pointers according to his/her attributes. Therefore, we collectively store all secret pointers accompanied by their multi-signature in a protected table isolated from other public tables. A legitimate DU can only read the entries that he/she is granted access to read. Moreover, the malicious outsider will not be able to see all the encrypted pointers and the malicious insider who can decrypt the outer layer of ciphertext will not be able to exploit the internal encrypted pointers to infer any information in the TDB.

Besides, we propose a complementary shuffling mechanism to further reduce the risk of any entity learning any information from the open data stored in the TDB. As demonstrated in the right part of Fig. \hyperref[fig:tdb]{6}, the TDB constantly shuffles the data to ensure that the order of the data is different each time the user accesses the TDB. Nevertheless, there is a possibility that the order of the data remains unchanged after the shuffle. If such a corner case occurs, the TDB will be automatically re-shuffled. This can be achieved by running a deterministic algorithm that compares the hash of the current data order with the hash of the previous data order. The algorithm returns $\bot$ when the shuffled data order is accidentally the same as the original data order. Thus, the TDB needs to reshuffle the data to avoid this problem. These will further enhance the security of TDB and leave attackers with no rules to follow.

\subsection{Signing Process}
We use multi-signature to place constraints between the SP and the DO. This allows the DO to confirm that the EHR obtained from the SP is valid. On the other side, the SP can ensure that the DO has not attempted to alter the original EHRs they provided. It is therefore possible to guarantee the integrity and authenticity of the EHR if they have reached an agreement to sign together on the \textit{same} message.

The following describes two issues we need to address when signing. Firstly, imagine a signature that is obtained by encrypting the hash of a message generated via a one-way hash function. This signature is said to be valid if the hash value generated by the verifier using the same hash function on the accompanying message is equivalent to the hash obtained by decrypting the signature provided by the signer. Such a signing and verification process establishes the integrity of the message but does not maintain its confidentiality since the message used to generate the hash is in its original form~\cite{An2002}. The second issue is how the SP and DO sign the same content when there are inconsistencies between the data held by the SP and DO after preprocessing the EHRs. To address these issues, we propose the following steps for signers to securely multi-sign the same content. A visualisation of the process is provided (see Fig. \hyperref[fig:signing]{7}).
\begin{itemize}
    \item \textbf{Step 1}: DO calls dataPreprocessing($\Phi$) to preprocess EHRs and encryptPointer($\mathfrak{pp}$,$\mathbb{M}$,$\mathbb{A}$,$\{k_{l}\}$) to encrypt the pointers with self-defined access policies.
    \item \textbf{Step 2}: DO sends the preprocessed EHRs with encrypted pointers to SP.
    \item \textbf{Step 3}: SP decrypts all encrypted pointers using decryptPointer($\mathfrak{pp}$,$\mathbb{C}$,$\mathcal{DK}_i$) and reconstructs the data by joining EHR blocks in the right order. When DO allows legitimate SPs with authorised attributes to decrypt all secrets, there should be no concern since the original data comes from the SP.
    \item \textbf{Step 4}: SP compares the reconstructed data with the original data maintained by itself. If they are identical, then the SP and DO have reached an agreement that the preprocessed EHRs have not been tampered with by the DO.
    \item \textbf{Step 5}: SP and DO interactively sign, using the algorithm multiSign($\mathcal{SK}_i$,$\mathbb{V}$,$\mathfrak{msg}$), on the hash of the confirmed EHR data obtained in step 2.
\end{itemize}

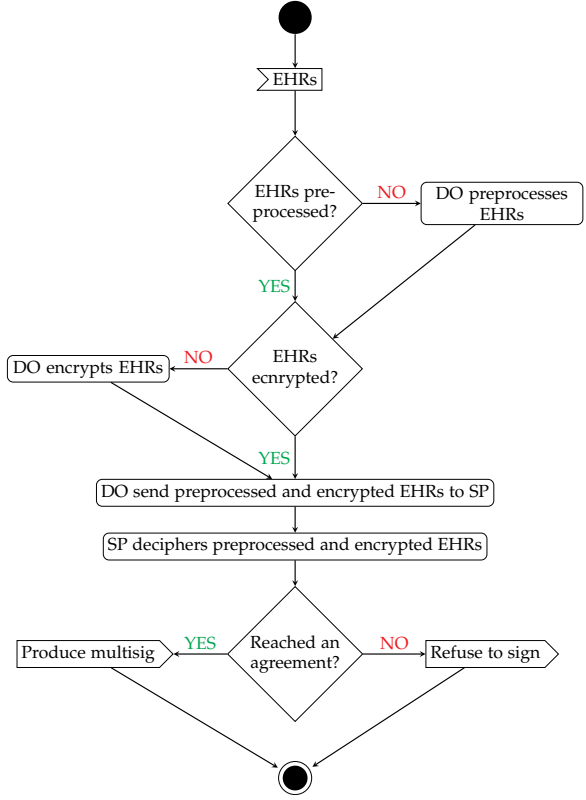
\begin{figure}[t]
    \centering
    \scalebox{0.55}{
        \begin{tikzpicture}
            [node distance=1.5cm,
                start/.style={circle, fill=black, minimum size=8mm},
                end/.style={path picture={\draw circle [radius=4mm]; \fill circle [radius=3mm];}},
                activity/.style={rectangle, draw, text centered, rounded corners},
                decision/.style={diamond, draw, text width=5em, text badly centered, inner sep=0pt,
                        node distance=2cm},
                input/.style={signal, draw, signal from=west, signal to=nowhere},
                output/.style={signal, draw, signal to=east},
                arrow/.style={thick, ->, >=stealth}]
            \node [start] (start) {};
            \node [font=\large, input, below of=start] (load_ehr) {EHRs};
            \node [font=\large, decision, below of=load_ehr, text width=7em, node distance=3cm] (preprocessed) {EHRs preprocessed?};
            \node [font=\large, decision, below of=preprocessed, text width=7em, node distance=4cm] (encrypted) {EHRs ecnrypted?};
            \node [font=\large, activity, right of=preprocessed, text width=11em, node distance=5cm] (preprocess_ehr) {DO preprocesses EHRs};
            \node [font=\large, activity, left of=encrypted, text width=11em, node distance=5cm] (encrypt_ehr) {DO encrypts EHRs};
            \node [font=\large, activity, below of=encrypted, node distance=3cm] (send_ehr) {DO send preprocessed and encrypted EHRs to SP};
            \node [font=\large, activity, below of=send_ehr, node distance=1.3cm] (decipher_ehr) {SP deciphers preprocessed and encrypted EHRs};
            \node [font=\large, decision, below of=decipher_ehr, text width=7em, node distance=2.6cm] (identical) {Reached an agreement?};
            \node [font=\large, output, right of=identical, node distance=4.6cm] (refuse) {Refuse to sign};
            \node [font=\large, output, left of=identical, node distance=5cm] (sign) {Produce multisig};
            \node [end, minimum size=8.2mm, below of=identical, node distance=3cm] (end) {};

            \draw [arrow] (start) -- (load_ehr);
            \draw [arrow] (load_ehr) -- (preprocessed);
            \draw [arrow] (preprocessed) -- node [anchor=east] {\large\textcolor{Green}{YES}} (encrypted);
            \draw [arrow] (preprocessed) -- node [anchor=south] {\large\textcolor{Red}{NO}} (preprocess_ehr);
            \draw [arrow] (preprocess_ehr) -- (encrypted);
            \draw [arrow] (encrypted) -- node [anchor=east] {\large\textcolor{Green}{YES}} (send_ehr);
            \draw [arrow] (encrypted) -- node [anchor=south] {\large\textcolor{Red}{NO}} (encrypt_ehr);
            \draw [arrow] (encrypt_ehr) -- (send_ehr);
            \draw [arrow] (send_ehr) -- (decipher_ehr);
            \draw [arrow] (decipher_ehr) -- (identical);
            \draw [arrow] (identical) -- node [anchor=south] {\large\textcolor{Green}{YES}} (sign);
            \draw [arrow] (identical) -- node [anchor=south] {\large\textcolor{Red}{NO}} (refuse);
            \draw [arrow] (sign) -- (end);
            \draw [arrow] (refuse) -- (end);
        \end{tikzpicture}}
    \caption{An illustration of how data owners and service providers can regulate each other to ensure the accuracy and integrity of the EHR.}
    \label{fig:signing}
\end{figure}

\section{Security Analysis}\label{section:5}
In this section, we analyse and prove the security of our proposed scheme formally against the adversarial model described in Section \hyperref[Adversarial]{3}. To ensure that E-Tenon is secure and resilient to a range of possible attacks, ML-ABE (a variant of CP-ABE) and MS-BN (a variant of Schnorr signature) are selected and integrated for reliability and validity. First, we note that ML-ABE is a proven CCA-1 secure scheme, where CCA-1 refers to the non-adaptive chosen-ciphertext attacks. Second, MS-BN is a proven secure scheme against the multi-user unforgeability against chosen message attacks (MU-UF-CMA). Our E-Tenon scheme should naturally inherit the security properties of these two building blocks. 

\begin{theorem}
Assume that the ML-ABE scheme in~\cite{Kaaniche2017} is selectively CCA-1 secure. Then, the E-Tenon system preserves confidentiality and is selectively CCA-1 secure with respect to the CCA-1 security game and Definition \hyperref[cca-1]{7}.
\end{theorem}

\begin{proof}
To prove the security of the E-Tenon system with respect to Definition \hyperref[cca-1]{7}, we consider there exist two polynomial-time adversaries $\mathcal{A}$ and $\mathcal{B}$, and a challenger $\mathcal{C}$. Here $\mathcal{B}$ is a simulator algorithm to run the security game defined in the naive CP-ABE. The security game $\mathcal{G}^{CCA-1}_{\mathcal{A}}(\lambda)$ is simulated as a non-adaptive chosen ciphertext attack against the proposed model by the adversary $\mathcal{A}$. It proceeds with $\mathcal{A}$, $\mathcal{B}, \text{and}\ \mathcal{C}$ in four phases as follows:
\begin{itemize}
    \item \textbf{Setup:} $\mathcal{C}$ runs setup algorithm with the security parameter $\lambda$ to obtain the public parameters and the master key $\mathfrak{msk},\mathfrak{pp}\leftarrow \text{setup}(\lambda)$, where $\mathfrak{msk}$ is defined as $(\delta,g^{\gamma})$ and $\mathfrak{pp}$ is defined as $\{\mathbb{G}_0,\mathbb{G}_1,p,g,g^{\delta},e,e(g,g)^{\gamma}\}$. Upon generation, $\mathcal{C}$ sends $\mathfrak{pp}$ to $\mathcal{B}$. Then $\mathcal{B}$ forward the same $\mathfrak{pp}$ to $\mathcal{A}$.
    \item \textbf{Query:} $\mathcal{B}$ initialises an empty table $T$, an integer session counter $j$ starting from zero and an empty set $\mathbb{Q}$. $\mathcal{A}$ can repeatedly query the following during this phase:
        \begin{itemize}
            \item \textbf{Create:} $\mathcal{B}$ asks $\mathcal{C}$ to increment $j$ by 1. $\mathcal{B}$ asks $\mathcal{C}$ to run the setup algorithm $\mathfrak{msk},\mathfrak{pp}\leftarrow \text{setup}(\lambda)$ and the keyGeneration algorithm $keys[\mathcal{DK}]\leftarrow \text{keyGeneration}(\mathfrak{pp},\mathfrak{msk},\mathbb{S})$ to extract a decryption key $\mathcal{DK}$ on $\mathbb{S}$ and the corresponding security levels $k_l$. Upon receiving $\mathcal{DK}$ from $\mathcal{C}$, $\mathcal{B}$ stores the entry $(j, \mathbb{S}, \mathfrak{pp}, \mathfrak{msk}, \mathcal{DK})$ in $T$ if it is not a duplicate entry and shares the decryption key $\mathcal{DK}$ with $\mathcal{A}$.
            \item \textbf{Corrupt:} $\mathcal{A}$ requests the decryption output of a ciphertext $\mathbb{C}$ using $\mathcal{DK}$ on $\mathbb{S}$. $\mathcal{B}$ checks if there is a previously extracted $\mathcal{DK}$ for $\mathbb{S}$ in the table $T$. If yes, $\mathcal{B}$ sets $\mathbb{Q} = \mathbb{Q}\cup\mathbb{S}$ and proceeds. Otherwise, $\mathcal{B}$ asks $\mathcal{C}$ to run the Create phase again and extract the corresponding $\mathcal{DK}$, such that the challenge access structure $\mathbb{A^*}(j,\mathbb{S},k_l)$ is equal to 1. 
            \item \textbf{Decrypt:} Upon receiving $\mathcal{DK}$, $\mathcal{B}$ decrypts the ciphertext $\mathbb{C}$ with $\mathcal{DK}$ using the decryption algorithm presented in the naive CP-ABE scheme. Finally, $\mathcal{B}$ returns the decryption output of the ciphertext $\mathbb{C}$ to $\mathcal{A}$.
        \end{itemize}
    \item \textbf{Challenge:} $\mathcal{A}$ chooses two plaintext messages $\mathbb{M}_0$ and $\mathbb{M}_1$ of the same length to be encrypted, which must remain unqueried until then. $\mathcal{A}$ also submits a challenge access structure $\mathbb{A}^*$ such that $\mathbb{S}$ does not satisfy $\mathbb{A}^*$ for all $\mathbb{S}\in\mathbb{Q}$. Upon receiving $\mathbb{A}^*$, $\mathcal{B}$ creates its own access structure $\mathbb{A}_\mathcal{B}$ based on the challenge access structure submitted by $\mathcal{A}$, such that $\mathbb{A}_\mathcal{B}\subseteq\mathbb{A}^*$. Next, $\mathcal{B}$ asks $\mathcal{C}$ to generate the ciphertext based on $\mathbb{M}_0$, $\mathbb{M}_1$ and $\mathbb{A}_\mathcal{B}$. $\mathbb{C}$ then randomly selects a bit $b\in\{0,1\}$ and outputs the encryption results of $\mathbb{M}_b$ under $\mathbb{A}_\mathcal{B}$ to $\mathcal{B}$. Finally, $\mathcal{B}$ forward the output to $\mathcal{A}$.
    \item \textbf{Guess:} $\mathcal{A}$ outputs its guess $b^\prime\in\{0,1\}$ for $b$. $\mathcal{A}$ wins the game if $b^\prime=b$.
\end{itemize}

In order to determine the adversary's advantage at this stage, some basic observations are necessary to be made. It is noted that the element $\tilde{\mathbb{C}}_{k_{l}}$ within the ciphertext encrypted by $\mathcal{C}$ during the challenge phase is either $\mathbb{M}_0 \cdot e(g, g)^{\mathrm{\gamma\varsigma}_{l}}$ or $\mathbb{M}_0 \cdot e(g, g)^{\mathrm{\gamma\varsigma}_{l}}$. Thus, the advantage for the adversary to distinguish between the two cases is $\mathrm{Adv}^{CCA-1}_{\mathcal{A}}(1^\lambda)\leq \epsilon$.
Now, let us take into account a modified game ${\mathcal{G}^{CCA-1}_{\mathcal{A}}}^\prime$. In this game, the main difference is that the element $\tilde{\mathbb{C}}_{k_{l}}$ of the challenge ciphertext becomes either $\mathbb{M}_0 \cdot e(g, g)^{\mathrm{\gamma\varsigma}_{l}}$ or $\mathbb{M}_1 \cdot e(g, g)^\theta$, where $\theta$ is chosen at random out of an additive group, $\theta${\scriptsize$\ \xleftarrow{\$}\ $}$\mathbb{O}_p$. Accordingly, the advantage of the adversary in winning the modified game becomes ${\mathrm{Adv}^{CCA-1}_{\mathcal{A}}}^\prime(1^\lambda)\geq \frac{1}{2} \cdot \epsilon$.
Then we simulate the attack over the modified security game based on case 1 of~\cite{Kaaniche2017}. A challenger $\mathcal{C}$ first chooses two exponents $\gamma$ and $\delta$ at random from $\mathbb{Z}_p$, such that $\gamma,\delta${\scriptsize$\ \xleftarrow{\$}\ $}$\mathbb{Z}_{p}$. $\mathcal{C}$ then obtains and shares the public parameters with the adversary in a special encoding: $\mathfrak{E}_{0}(1)=g$, $\mathfrak{E}_{0}(\delta)=g^{\delta}$ and $\mathfrak{E}_{T}(\gamma)$. In the subsequent challenge phase, the adversary $\mathcal{A}$ again asks challenger $\mathcal{C}$ to encrypt the challenge message under the access structure ${\mathbb{A}^{\prime}}^*$. After that, the adversary $\mathcal{A}$ gets $\mathbb{C}_{k_{l}} = g^{\delta \varsigma_{l}}$ and $\tilde{\mathbb{C}}_{k_{l}} = e(g^{\delta}, g^{\delta})^{\theta_l}$ for each defined security level along with the relevant attributes. It is worth pointing out that the request from adversary $\mathcal{A}$ will not be granted if $\mathcal{A}$ requests a set of attributes that can satisfy all the security levels defined in the challenge access structure. In other contradictory cases, the game terminates immediately, and the adversary loses the game. Finally, we use the big-$\mathcal{O}$ notation to express the upper limit of the adversary's advantage in winning the aforementioned security game as ${\mathrm{Adv}^{CCA-1}_{\mathcal{A}}}^\prime(1^\lambda)\leq \mathcal{O}(\frac{c^*\cdot q^2}{p})$, where $c^*$ is the bound on the maximum number of security level can be set, $q$ is the bound on the maximum number of group elements obtained by $\mathcal{A}$, and $p$ is the order of an additive group $\mathbb{O}_p$. Hence, we state that the proposed E-Tenon system is CCA-1 secure and the confidentiality of EHR is guaranteed under the Generic Group Model if no PPT adversary can selectively break the security naive CP-ABE and ML-ABE with non-negligible advantage.
\end{proof}

\begin{theorem}
Assume that the ML-ABE scheme in~\cite{Kaaniche2017} is private against both malicious and honest-but-curious adversaries. Then, the proposed E-Tenon system preserves privacy against both malicious DU and honest-but-curious TDB.
\end{theorem}
\begin{proof}
In this proof, we consider attacks from a malicious DU and an honest-but-curious TDB, respectively. First of all, it is worth noting that the malicious adversary DU will have the same advantage as in $\mathcal{G}^{CCA-1}_{\mathcal{A}}(\lambda)$ when a DU tries to extend or override his/her access rights to gain additional access to the encrypted information (e.g., the embedded enciphering secret $\varsigma_{l}$). This is because such a scenario is in line with the confidentiality property. Next, let us recall that the secret relationships $\{{\mathfrak{N}_{l}}^*\}_{l\in \{1,c^*\}}$ in the ciphertext are independently encrypted with a set of different security levels $\{{k_{l}}^*\}_{l\in \{1,c^*\}}$ thanks to the use of multi-level ABE. Thus, in order to deduce any information from any part of a challenge ciphertext, or to break the indistinguishability property, the adversary DU must be able to recover $e(g, g)^{\mathrm{\gamma\varsigma}_{l}}$ together with the corresponding $\tilde{\mathbb{C}}_{k_{l}}=\mathfrak{N}_i \cdot e(g, g)^{\mathrm{\gamma\varsigma}_{l}}$ and $\mathcal{D} =g^{\frac{\gamma +\mathfrak{r}}{\delta }}$. However, the proof of Theorem 1 shows that the adversary only has a negligible advantage in selectively breaking the CCA-1 security of E-Tenon. Our framework, therefore, prevents malicious DUs from revealing any information, as ML-ABE does not disclose any useful information.

In another scenario, let us assume that the honest-but-curious TDB complies with its obligations. However, it tries to reveal which DO upload the EHR or which DU requested to retrieve the EHR. This clearly compromises the privacy property. Having said that, we show that the TDB does not have the ability to distinguish requesters by their attributes. Suppose $DO_x$ and $DO_y$ are two patients with a set of distinct attributes in the proposed system. Their $\mathbb{A}$ will be indistinguishable as ML-ABE inherits such property from the naive CP-ABE scheme, such that $\mathbb{A}(\mathbb{S}_{DO_x})=1$ and $\mathbb{A}(\mathbb{S}_{DO_y})=1$ for $\mathbb{S}_{DO_x} \neq \mathbb{S}_{DO_y}$. Therefore, the honest-but-curious TDB is unable to identify DOs and DUs. Hence, our system is secure against both internally and externally launched attacks.
\end{proof}

\begin{theorem}
Assume that the MS-BN scheme in~\cite{Bellare2006} is MU-UF-CMA secure. Then the proposed E-Tenon system is MU-UF-CMA secure with respect to the MU-UF-CMA security game and Definition \hyperref[uf-cma]{8}.
\end{theorem}
\begin{proof}
Let $\mathcal{F}$ be a PPT adversary running in time at most $t$ against the multi-signature algorithm. Let $q_p$ and $N$ denote the number of signing processes initiated by $\mathcal{F}$ and the number of verification keys in the set $\mathbb{V}$, respectively, and let $q_r$ be the maximum number of random oracle queries that $\mathcal{F}$ can make. 

As proved in~\cite{Bellare2006}, breaking the MS-BN model is considered to be at least as hard as the discrete logarithm problem (DLP) for an adversary $\mathcal{F}$ under the random oracle model (ROM). Below we recapitulate several important points discussed by Bellare and Neven based on their Forking Lemmas. Firstly, the accepting probability acc and the forking probability frk of $\mathcal{F}$ used in their General Forking Lemma are quantified as follows:
\begin{align*}
\mathrm{frk} &\geq \operatorname{acc} \cdot\left(\frac{\mathrm{acc}}{q}-\frac{1}{h}\right) \\
\mathrm{acc} & \geq \epsilon-\frac{\left(q_{r}+N \cdot q_{p}+1\right)^{2}}{2^{l_{0}}}-\frac{2 q_{p}\left(q_{r}+N \cdot q_{p}\right)}{2^{k}}
\end{align*}
Then, we square of the acceptance rate acc, which gives us the $\mathrm{acc}^{2}$ as below:
\begingroup
\allowdisplaybreaks
\begin{align*}
\mathrm{acc}^{2} & \geq \left( \epsilon -\frac{( q_{r} +N\cdot q_{p} +1)^{2}}{2^{l_{0}}} -\frac{2q_{p}( q_{r} +N\cdot q_{p})}{2^{k}}\right)^{2}\\
 & \geq \epsilon ^{2} -\frac{\epsilon ( q_{r} +N\cdot q_{p} +1)^{2}}{2^{l_{0}}} -\frac{\epsilon \cdot 2q_{p}( q_{r} +N\cdot q_{p})}{2^{k}}\\
 & \ \ \ \ \ \ -\frac{\epsilon ( q_{r} +N\cdot q_{p} +1)^{2}}{2^{l_{0}}} +\frac{( q_{r} +N\cdot q_{p} +1)^{4}}{\left( 2^{l_{0}}\right)^{2}}\\
 & \ \ \ \ \ \ +\frac{( q_{r} +N\cdot q_{p} +1)^{2}}{2^{l_{0}}} \cdot \frac{2q_{p}( q_{r} +N\cdot q_{p})}{2^{k}}\\
 & \ \ \ \ \ \ -\frac{\epsilon \cdot 2q_{p}( q_{r} +N\cdot q_{p})}{2^{k}} +\left(\frac{2q_{p}( q_{r} +N\cdot q_{p})}{2^{k}}\right)^{2}\\
 & \ \ \ \ \ \ +\frac{2q_{p}( q_{r} +N\cdot q_{p})}{2^{k}} \cdot \frac{( q_{r} +N\cdot q_{p} +1)^{2}}{2^{l_{0}}}\\
 & \geq \epsilon ^{2} -\frac{2\epsilon ( q_{r} +N\cdot q_{p} +1)^{2}}{2^{l_{0}}} -\frac{4\epsilon \cdot q_{p}( q_{r} +N\cdot q_{p})}{2^{k}}\\
 & \geq \epsilon ^{2} -\frac{2( q_{r} +N\cdot q_{p} +1)^{2}}{2^{l_{0}}} -\frac{4q_{p}( q_{r} +N\cdot q_{p})}{2^{k}}
\end{align*}
\endgroup
If there exists an adversary $\mathcal{F}$ who manages to win the game $\mathcal{G}^{ROM}_{\mathcal{F}}(t,q_p,q_r,N,\epsilon)$, then it implies that there is an adversary $\mathcal{F}^\prime(\epsilon^\prime,t^\prime)$ that can solve the DLP. Thus, the probability $\epsilon^\prime$ of adversary $\mathcal{F}^\prime$ successfully solving the DLP and the corresponding running time $t^\prime$ for $\mathcal{F}^\prime$ to solve the DLP are given by:
\begingroup
\allowdisplaybreaks
\begin{align*}
t^{\prime}= & \; 2 t+q_{p} t_{\mathrm{exp}}+\mathcal{O}\left(\left(q_{p}+q_{r}\right)\left(1+q_{r}+N q_{p}\right)\right)\\
\epsilon^{\prime}  \geq & \; \mathrm{ frk } \\
\geq & \operatorname{acc} \cdot\left(\frac{\mathrm{acc}}{q}-\frac{1}{h}\right)\\
\geq &\frac{\mathrm{acc}^{2}}{q}-\frac{\mathrm{acc}}{h} \\
\geq &\frac{\mathrm{acc}^{2}}{q}-\frac{1}{2^{l_{1}}} \\
\geq &\frac{\epsilon^{2}-\frac{2\left(q_{r}+N \cdot q_{p}+1\right)^{2}}{2^{l_{0}}}-\frac{4 q_p\left(q_r+N \cdot q_p\right)}{2^{k}}}{q_r+q_p}-\frac{1}{2^{l_{1}}} \\
\geq &\frac{\epsilon^{2}}{q_r+q_p}-\frac{2 q_r+16 N^{2} \cdot q_p}{2^{l_{0}}}-\frac{8 N\cdot q_p}{2^{k}}-\frac{1}{2^{l_{1}}}
\end{align*}
\endgroup
Here, $t^{\prime}$ is two times the running time $t$ required by $\mathcal{F}$ plus the time required to solve the DLP. One can argue that if there is no algorithm capable of solving the DLP. Then there is no adversary capable of breaking the security of MS-BN with any reasonable probability. Therefore, the proposed E-Tenon system is also MU-UF-CMA secure against integrity and authenticity attacks by inheriting the security properties of the Multi-Signature scheme MS-BN.
\end{proof}

\section{Performance Evaluation}\label{section:6}
In this section, we discuss the performance of the proposed model. We first compare our scheme with other competitive solutions in terms of security properties. We then evaluate the relevant computation cost of the E-Tenon in different tasks. Subsequently, we discuss the communication and storage costs of E-Tenon.

\begin{table}[ht]
    \centering
    \caption{Security properties and functionalities comparison with related works.}
    \scalebox{0.85}{
        \begin{tabular}{|c|c|c|c|c|c|c|c|c|c|c|}
            \hline
            \textbf{}                      & \textbf{SP1} & \textbf{SP2} & \textbf{SP3} & \textbf{SP4} & \textbf{SP5} & \textbf{SP6} & \textbf{SP7} & \textbf{SP8} & \textbf{SP9} & \textbf{SP10}                                                                                                                                                                                                                                                     \\

            \hline\hline
            \cite{Green2011}               & \ding{55}    & \ding{51}    & \ding{51}    & \ding{55}    & \ding{55}    & \ding{51}    & \ding{51}    & \ding{55}    & \ding{51}    & \LEFTcircle                                                                                                                                                                                                                                                       \\

            \hline
            \cite{Sun2018_EHR}             & \ding{55}    & \ding{51}    & \ding{51}    & \ding{51}    & \ding{51}    & \LEFTcircle  & \ding{51}    & \ding{55}    & \ding{51}    & \ding{55}                                                                                                                                                                                                                                                         \\

            \hline
            \cite{Camenisch2015}           & \ding{55}    & \ding{51}    & \ding{51}    & \ding{51}    & \ding{51}    & \ding{51}    & -            & \ding{55}    & \ding{55}    & \ding{55}                                                                                                                                                                                                                                                         \\

            \hline
            \cite{10.1145/2897845.2897870} & \ding{55}    & \ding{55}    & \ding{51}    & \ding{55}    & \ding{55}    & \ding{51}    & \LEFTcircle  & \ding{55}    & \ding{51}    & \LEFTcircle                                                                                                                                                                                                                                                       \\

            \hline
            \cite{Zhang2018_EHR}           & \ding{55}    & \ding{51}    & \ding{51}    & \ding{51}    & \ding{51}    & \ding{51}    & \ding{51}    & \ding{55}    & \ding{51}    & \ding{51}                                                                                                                                                                                                                                                         \\

            \hline
            \cite{Maffei2015}              & \ding{55}    & \ding{51}    & \ding{51}    & \ding{51}    & \ding{51}    & \ding{51}    & \ding{51}    & \ding{55}    & \ding{51}    & \ding{51}                                                                                                                                                                                                                                                         \\

            \hline
            Ours                           & \ding{51}    & \ding{51}    & \ding{51}    & \ding{51}    & \ding{51}    & \ding{51}    & \ding{51}    & \ding{51}    & \ding{51}    & \ding{51}                                                                                                                                                                                                                                                         \\
            \hline

            \multicolumn{11}{|c|}{\multirow{6}{*}{\begin{tabular}[c]{@{}c@{}}
                                                              \ding{51}: Fully Satisfied \ \ \ \ \ \ding{55}: Not Satisfied \ \ \ \ \ \LEFTcircle: Partially Satisfied\ \ \ \ \ -: N/A \\
                                                              \textbf{SP1}: Open Database; \textbf{SP2}: Secure-channel Free; \textbf{SP3}: Data Confidentiality;                      \\\textbf{SP4}: Data Integrity; \textbf{SP5}: Non-Repudiation; \textbf{SP6}: User Privacy; \\\textbf{SP7}: Collusion Resistance; \textbf{SP8}: Multi-level Access Control; \\\textbf{SP9}: Fine-grained Access Control; \textbf{SP10}: Process Transparency\end{tabular}}} \\
            \multicolumn{11}{|c|}{}                                                                                                                                                                                                                                                                                                                                                                                                                   \\
            \multicolumn{11}{|c|}{}                                                                                                                                                                                                                                                                                                                                                                                                                   \\
            \multicolumn{11}{|c|}{}                                                                                                                                                                                                                                                                                                                                                                                                                   \\
            \multicolumn{11}{|c|}{}                                                                                                                                                                                                                                                                                                                                                                                                                   \\
            \multicolumn{11}{|c|}{}                                                                                                                                                                                                                                                                                                                                                                                                                   \\ \hline
        \end{tabular}}
    \label{properties}
\end{table}

\begin{table*}[ht]
    \centering
    \caption{Performance benchmarking based on computation, communication and storage cost.}
    \label{costs}
    \scalebox{1}{
        \begin{tabular}{|c|c|c|c|c|}
            \hline
            \textbf{Cost}                        & \multicolumn{2}{|c|}{\textbf{E-Tenon}}                                                                                                                                                                                                           & \multicolumn{2}{|c|}{\textbf{Zhang et al.}\cite{Zhang2018_EHR}}                                                                                                                                                                                                                                                                                                                                                                                                                                                                                                                         \\

            \hline\hline
            \textbf{Signing}                     & $\mathcal{T}_\mathbf{exp}$                                                                                                                                                                                                                       & $\approx 2.34\ ms$                                              & $3\mathcal{T}_\mathbf{mult}$                                                                                                                                                                                                                   & $\approx 43.50\ ms$                                                                                                                                                                                                                                                  \\
            \hline
            \textbf{Verification}                & $\mathcal{T}_\mathbf{exp}$                                                                                                                                                                                                                       & $\approx 2.34\ ms$                                              & $\mathcal{T}_\mathbf{mult}+3\mathcal{T}_\mathbf{par}$                                                                                                                                                                                          & $\approx 25.84\ ms$                                                                                                                                                                                                                                                  \\
            \hline
            \textbf{Encryption}                  & $k\mathcal{T}_\mathbf{mult}+2(k+\mathfrak{l}_{MST})\mathcal{T}_\mathbf{exp}$                                                                                                                                                                     & $\approx 142.70 \ ms$                                           & $k\mathcal{T}_\mathbf{mult}+2 k(1+\mathfrak{l}_{AT})\mathcal{T}_\mathbf{exp}$                                                                                                                                                                  & $\approx 329.90 \ ms$                                                                                                                                                                                                                                                \\
            \hline
            \multirow{2}{*}{\textbf{Decryption}} & \multirow{2}{*}{ $\begin{array}{l}(\mathfrak{n}_{MST}+\mathfrak{l}_{AT})(2\mathcal{T}_\mathbf{par}+\mathcal{T}_\mathbf{exp}+\mathcal{T}_\mathbf{mult})\\+\mathcal{T}_\mathbf{mult}(2+m\mathfrak{n}_{MST})+\mathcal{T}_\mathbf{par}\end{array}$ } & \multirow{2}{*}{$\approx 761.28 \ ms$}                          & \multirow{2}{*}{$\begin{array}{l}(k\mathfrak{n}_{AT}+\mathfrak{l}_{AT})(2\mathcal{T}_\mathbf{par}+\mathcal{T}_\mathbf{exp}+\mathcal{T}_\mathbf{mult}) \\+\mathcal{T}_\mathbf{mult}(2+m\mathfrak{n}_{AT})+\mathcal{T}_\mathbf{par}\end{array}$} & \multirow{2}{*}{$\approx 1249.28 \ ms$}                                                                                                                                                                                                                              \\
                                                 &                                                                                                                                                                                                                                                  &                                                                 &                                                                                                                                                                                                                                                &                                                                                                                                                                                                                                                                      \\
            \hline
            \textbf{Signature}                   & $2|{ecc}|$                                                                                                                                                                                                                                       & $\approx 320 \ bits $                                           & $3|{ecc}|$                                                                                                                                                                                                                                     & $\approx 480 \ bits$                                                                                                                                                                                                                                                 \\
            \hline
            \textbf{Ciphertext}                  & $\{\left|MST\right|,2(k+\mathfrak{l}_{MST})\left|\mathbb{G}\right|\}$                                                                                                                                                                            & $\approx 3.86 \ kb$                                             & $\{k\left|AT\right|,2k(1+\mathfrak{l}_{AT})\left|\mathbb{G}\right|\}$                                                                                                                                                                          & $\approx 14.18 \ kb$                                                                                                                                                                                                                                                 \\
            \hline
            \multicolumn{5}{|c|}{\multirow{5}{*}{\begin{tabular}[c]{@{}c@{}}$\mathcal{T}_\mathbf{exp}$: cost of a modular exponentiation (2.34 ms); $\mathcal{T}_\mathbf{mult}$: cost of a multiplication (14.5 ms); $\mathcal{T}_\mathbf{par}$: cost of a bilinear pairing (3.78 ms);\\ $\mathfrak{l}$: number of external nodes - i.e. attributes in the tree (10); $\mathfrak{n}$: number of internal nodes - i.e. threshold gates in the tree (5); $|{ecc}|$: size of \\the elliptic curve (160 bits); $m$: number of child nodes of the threshold gates (5); $|MST|$: size of an aggregate access tree (160 bits);\\ $|AT|$: size of a separate access tree (160 bits); $k$: number of access tree (5); $|\mathbb{G}|$: bit length of the element in the group (1024 bits) \end{tabular}}} \\
            \multicolumn{5}{|c|}{}                                                                                                                                                                                                                                                                                                                                                                                                                                                                                                                                                                                                                                                                                                                                                                                                                                                            \\
            \multicolumn{5}{|c|}{}                                                                                                                                                                                                                                                                                                                                                                                                                                                                                                                                                                                                                                                                                                                                                                                                                                                            \\
            \multicolumn{5}{|c|}{}                                                                                                                                                                                                                                                                                                                                                                                                                                                                                                                                                                                                                                                                                                                                                                                                                                                            \\
            \multicolumn{5}{|c|}{}                                                                                                                                                                                                                                                                                                                                                                                                                                                                                                                                                                                                                                                                                                                                                                                                                                                            \\ \hline
        \end{tabular}
    }
\end{table*}

\subsection{Security Properties}
To compare security properties and functionalities, we have selected several state-of-the-art schemes (\cite{Sun2018_EHR,Green2011,Camenisch2015,Maffei2015,10.1145/2897845.2897870,Zhang2018_EHR}) for protecting EHRs and compared them on various dimensions. A summarised comparison of the security properties and characteristics of the schemes is presented in Table \hyperref[properties]{2}. Although there are wide-ranging interesting solutions, they still suffer from different shortcomings and do not work efficiently where open databases are concerned. The scheme proposed by Sun et al.~\cite{Sun2018_EHR} employs attribute-based techniques, but the patient's involvement in the encryption and signing of the data is weakened. In~\cite{Sun2018_EHR}, the patient does not have the right to specify the access policy of their own data. Also, the doctor handles the encryption and signing process, meaning that the direct control of the data is entirely in the hands of the doctor, rather than the patient. Such a design increases the advantage for malicious insiders and makes the system less trustworthy for patients. In contrast, our E-Tenon system inherently gives more control to the patients since they are the actual owner of the EHR. In this way, they can set different levels of access policies for different types of data on their own, and they are allowed to engage in the process of Multi-Signature.

Green et al.~\cite{Green2011} have attempted to reduce the user's computational overheads by outsourcing the task of decryption to an untrusted cloud service provider (CSP). In their system, the CSP transform the ciphertext of ABE into a simple El Gamal-style ciphertext based on a transformation key provided by the data user. Despite the converted ciphertext requiring lower computational cost than its initial form when recovering the plaintext, the user cannot verify that the CSP has performed the transformation operation honestly. Similarly, the scheme presented in~\cite{Camenisch2015} ensures unlinkability of the stored data by converting identifying attributes into non-sensitive pseudonyms. However, this process is not transparent, meaning the data owner cannot audit their data flow. By comparison, the data pre-processing algorithm in our system is run on the data owner's side, and there is no need for other central entities to perform any secondary processing of the uploaded EHRs. Besides, instead of using a basic form of digital signature, we utilise multi-signature technology, which allows a group of participants to co-sign the same message effectively. This naturally enables the patient (DO) and the service provider (SP) to restrain each other's dishonest behaviour. Thus, it further enhances integrity, authenticity and non-repudiation. In this regard, we emphasise that multi-signature is more promising than the standard digital signature or other techniques that involve many signatures. Because in the absence of multiple entities constraining each other, the entity accessing the EHR later can replace the EHR provided by the previous entity and continue to sign the EHR supplied by itself. Therefore, the convectional signature approach does not guarantee the authenticity of the EHR in a collaborative environment.

\begin{figure*}
    \centering
    \subfloat[Signing]{%
        \scalebox{0.65}{
            \begin{tikzpicture}
                \begin{axis}[
                        xlabel={Number of Signing Tasks},
                        ylabel={Computation Cost (\textit{ms})},
                        xmin=10, xmax=50,
                        ymin=0, ymax=2500,
                        xtick={10,20,30,40,50},
                        ytick={0,500,1000,1500,2000,2500},
                        legend pos=north west,
                        ymajorgrids=true,
                        grid style=dashed,
                    ]

                    \addplot[
                        color=blue,
                        mark=triangle,
                    ]
                    coordinates {
                            (10,23.4)(20,46.8)(30,70.2)(40,93.6)(50,117)
                        };

                    \addplot[
                        color=red,
                        mark=square,
                    ]
                    coordinates {
                            (10,435)(20,870)(30,1305)(40,1740)(50,2175)
                        };

                    \legend{E-Tenon, Zhang et al.\cite{Zhang2018_EHR}}

                \end{axis}
            \end{tikzpicture}}}
    \subfloat[Verification]{%
        \scalebox{0.65}{
            \begin{tikzpicture}
                \begin{axis}[
                        xlabel={Number of Verification Tasks},
                        ylabel={Computation Cost (\textit{ms})},
                        xmin=10, xmax=50,
                        ymin=0, ymax=1400,
                        xtick={10,20,30,40,50},
                        ytick={0,200,400,600,800,1000,1200,1400},
                        legend pos=north west,
                        ymajorgrids=true,
                        grid style=dashed,
                    ]

                    \addplot[
                        color=blue,
                        mark=triangle,
                    ]
                    coordinates {
                            (10,23.4)(20,46.8)(30,70.2)(40,93.6)(50,117)
                        };

                    \addplot[
                        color=red,
                        mark=square,
                    ]
                    coordinates {
                            (10,258.4)(20,516.8)(30,775.2)(40,1033.6)(50,1292)
                        };

                    \legend{E-Tenon, Zhang et al.\cite{Zhang2018_EHR}}

                \end{axis}
            \end{tikzpicture}}}
    \subfloat[Encryption]{%
        \scalebox{0.65}{
            \begin{tikzpicture}
                \begin{axis}[
                        xlabel={Number of Attributes},
                        ylabel={Computation Cost (\textit{ms})},
                        xmin=10, xmax=50,
                        ymin=0, ymax=1400,
                        xtick={10,20,30,40,50},
                        ytick={0,200,400,600,800,1000,1200,1400},
                        legend pos=north west,
                        ymajorgrids=true,
                        grid style=dashed,
                    ]

                    \addplot[
                        color=blue,
                        mark=triangle,
                    ]
                    coordinates {
                            (10,142.7)(20,189.5)(30,236.29)(40,283.1)(50,329.9)
                        };

                    \addplot[
                        color=red,
                        mark=square,
                    ]
                    coordinates {
                            (10,329.9)(20,563.9)(30,797.9)(40,1031.9)(50,1265.89)
                        };

                    \legend{E-Tenon, Zhang et al.\cite{Zhang2018_EHR}}

                \end{axis}
            \end{tikzpicture}}}
    \\
    \subfloat[Decryption]{%
        \scalebox{0.65}{
            \begin{tikzpicture}
                \begin{axis}[
                        scaled y ticks = false,
                        xlabel={Number of Decryption Tasks},
                        ylabel={Computation Cost (\textit{ms})},
                        xmin=10, xmax=50,
                        ymin=0, ymax=70000,
                        xtick={10,20,30,40,50},
                        ytick={0,10000,20000,30000,40000,50000,60000,70000},
                        legend pos=north west,
                        ymajorgrids=true,
                        grid style=dashed,
                    ]

                    \addplot[
                        color=blue,
                        mark=triangle,
                    ]
                    coordinates {
                            (10,7612.8)(20,15225.6)(30,22838.4)(40,30451.2)(50,38064)
                        };

                    \addplot[
                        color=red,
                        mark=square,
                    ]
                    coordinates {
                            (10,12492.8)(20,24985.6)(30,37478.4)(40,49971.2)(50,62464)
                        };

                    \legend{E-Tenon, Zhang et al.\cite{Zhang2018_EHR}}

                \end{axis}
            \end{tikzpicture}}}
    \subfloat[Signature]{%
        \scalebox{0.65}{
            \begin{tikzpicture}
                \begin{axis}[
                        scaled y ticks = false,
                        xlabel={Number of Signatures},
                        ylabel={Communication and Storage Cost (\textit{bit})},
                        xmin=10, xmax=50,
                        ymin=0, ymax=24000,
                        xtick={10,20,30,40,50},
                        ytick={0,4000,8000,12000,16000,20000,24000},
                        legend pos=north west,
                        ymajorgrids=true,
                        grid style=dashed,
                    ]

                    \addplot[
                        color=blue,
                        mark=triangle,
                    ]
                    coordinates {
                            (10,3200)(20,6400)(30,9600)(40,12800)(50,16000)
                        };

                    \addplot[
                        color=red,
                        mark=square,
                    ]
                    coordinates {
                            (10,4800)(20,9600)(30,14400)(40,19200)(50,24000)
                        };

                    \legend{E-Tenon, Zhang et al.\cite{Zhang2018_EHR}}

                \end{axis}
            \end{tikzpicture}}}
    \subfloat[Ciphertext]{%
        \scalebox{0.65}{
            \begin{tikzpicture}
                \begin{axis}[
                        xlabel={Number of Ciphertexts},
                        ylabel={Communication and Storage Cost (\textit{kb})},
                        xmin=10, xmax=50,
                        ymin=0, ymax=800,
                        xtick={10,20,30,40,50},
                        ytick={0,100,200,300,400,500,600,700,800},
                        legend pos=north west,
                        ymajorgrids=true,
                        grid style=dashed,
                    ]

                    \addplot[
                        color=blue,
                        mark=triangle,
                    ]
                    coordinates {
                            (10,38.6)(20,77.2)(30,115.8)(40,154.4)(50,193)
                        };

                    \addplot[
                        color=red,
                        mark=square,
                    ]
                    coordinates {
                            (10,141.8)(20,283.6)(30,425.4)(40,567.2)(50,709)
                        };

                    \legend{E-Tenon, Zhang et al.\cite{Zhang2018_EHR}}

                \end{axis}
            \end{tikzpicture}}}
    \caption{Performance comparison based on computation, communication and storage cost (using simulation parameters specified in Table \hyperref[costs]{3}).}
    \label{fig8}
\end{figure*}
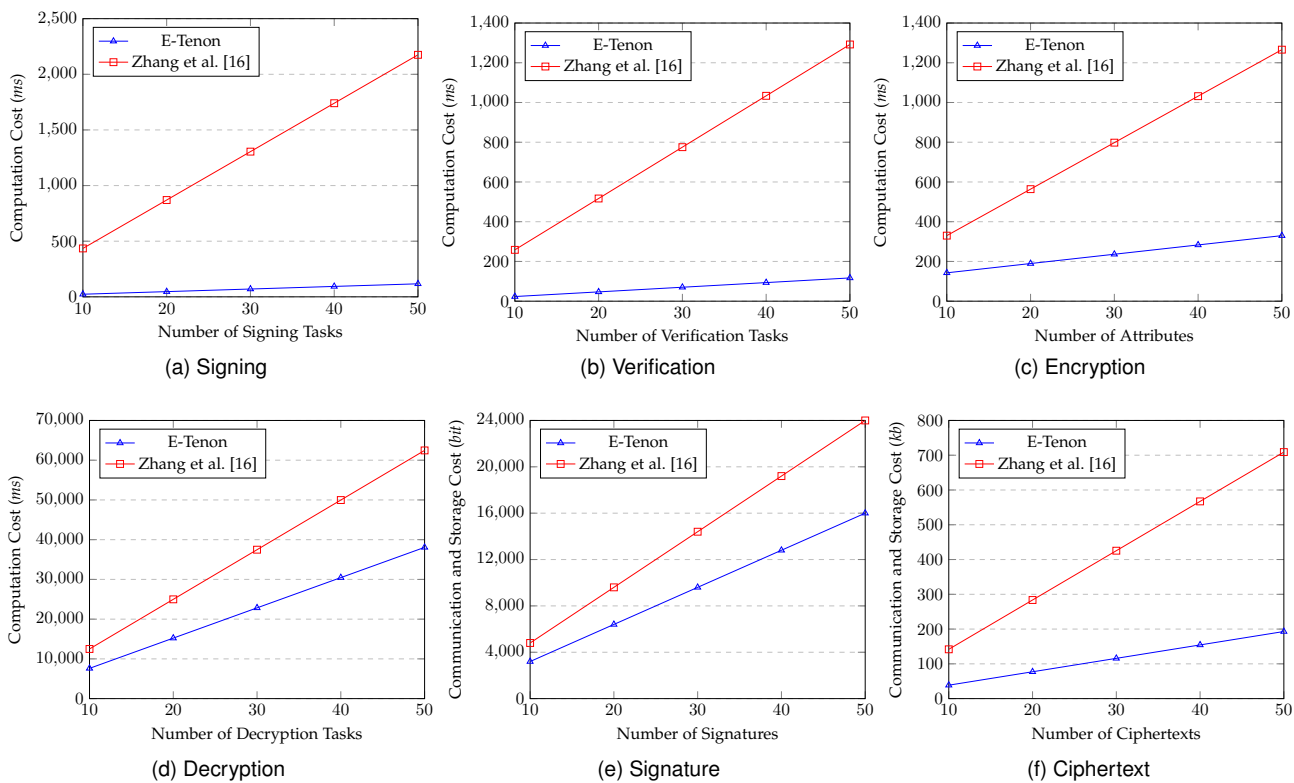

Furthermore, Huang et al.'s solution \cite{10.1145/2897845.2897870} focuses specifically on EHR confidentiality, although their solution is not security channel-free and we found no discussion of how they ensure EHR integrity, which makes their solution slightly less than perfect in our comparison. However, system proposed by Zhang et al.~\cite{Zhang2018_EHR} (SSH) and system proposed by Maffei et al.~\cite{Maffei2015} (GORAM) satisfied most of the security properties. GORAM allows data owners to share their data stored in the cloud selectively, and the storing entity is not permitted to inspect any data. Nevertheless, the strong security they have achieved comes at the cost of increasing the ciphertext size and slowing down encryption and decryption.

Finally, we observe that none of those mentioned above schemes can be applied to public databases where most of the data is stored in plaintext, and none of the encryption methods used in these schemes can efficiently implement multi-level access control. On the contrary, thanks to the novel concept of E-Tenon, our data is securely stored in an open database (TDB), which means the computational overhead on encryption and decryption is minimal compared to solutions based on heavy encryption.

\subsection{Computation Cost}
We use a virtual machine (Ubuntu 12.04) with an Intel Core i5-4200M dual-core 2.50 GHz CPU to conduct simulations of the core operations based on three main libraries: JPBC library Pbc-05.14~\cite{pbc0514}, JCE library~\cite{jce} and Apache Commons IO library~\cite{uuid}. We test modular exponentiation,  multiplication and bilinear pairing for 2,000 times and take the average CPU time in milliseconds. Table \hyperref[costs]{3} shows the cost for signing, verification, and encryption and decryption. Firstly, the signing and verification algorithms adapted in our model outperform other the relevant algorithms in the state-of-the-art schemes~\cite{sig_Lu,Boldyreva2002,Zhang2018_EHR}. This is because there is only one exponentiation operation required when an entity signs/verifies the message (the average CPU time for 2000 trials is approximately equal to only 2.34 $ms$). In addition, since it is a practical requirement to protect different types of EHR data according to different levels of security, our system uses ML-ABE's aggregated master access structure to meet this requirement effectively. It is worth noting that the schemes built on the classic CP-ABE (e.g. \cite{Zhang2018_EHR}) need to create a separate access structure for each defined security level $\{k_{l}\}_{l\in \{1,c\}}$ in order to achieve the same security functionality as we have. However, using multiple access structures will inevitably create many duplicate attributes. So our system saves computational overhead by avoiding duplicate nodes and unnecessary polynomials in the access structure, such that $\sum _{l=1}^{c}\mathfrak{l}_{AT_{k_{l}}} \geq \mathfrak{l}_{MST}$ ($\mathfrak{l}$ denotes the number of attributes/external nodes). A more intuitive comparison of performance with \cite{Zhang2018_EHR} is visualised in Fig. \hyperref[fig8]{8}. Furthermore, the advantages of our approach can also be seen in the following scenario. It is common knowledge that the size of EHR can vary from a few bits to tens or even hundreds of megabytes (e.g., 100 bits - 100 MB). However, we are only encrypting relationships between different EHR blocks, that is, instead of encrypting the whole EHR data, we only encrypt a number of constant sized pointers (16 Bytes). This idea reduces the time taken for encryption and decryption considerably, thanks to the use of the \emph{electronic tenon structure}.

\subsection{Communication and Storage Cost}
Finally, we analyse the communication and storage costs of the proposed protocol. As mentioned above, the access structure used by E-Tenon is designed in an aggregated manner, and the cost of our scheme in terms of communication and storage is optimised by eliminating duplicate attributes. This implies that the size of the ciphertext in the E-Tenon system is shorter than other schemes with a series of separate access structures.
However, our protocol requires an extra round of communication during the signing process as compared to other schemes, which is a trade-off for supporting concurrent signing in the multi-user environment, as pointed out in MS-BN~\cite{Bellare2006}. That being said, the size of our signature is only $2|{ecc}|$ (note that different schemes may work over a different n-bit elliptic curve). Following the security discussion in~\cite{Koblitz2000}, the use of a 160-bit elliptic curve would provide about the equivalent security level as DSA (Digital Signature Algorithm) and RSA (Rivest–Shamir–Adleman) with a 1024-bit modulus. Therefore, let us assume that we currently require the same level of security as stated above. The size of the multiple signature $\sigma$ is only 320 bits (40 Bytes) in this case. Taken together, the discussion suggests that we have achieved more secure and reliable protection of EHR without compromising efficiency.

\section{Conclusion and Future Work}\label{section:7}
This paper proposed an efficient privacy-preserving open data sharing scheme for a secure EHR system. The idea of keeping most of the data open without compromising security and privacy is considered as a novel attempt in this field. Moreover, we presented in detail the effective integration of two promising technologies in our E-Tenon system: ML-ABE and Multi-Signature in the direction of protecting security of EHR and patient privacy. Our solution exploits the advantages of ABE for key management and multiple signatures for protecting the authenticity and integrity of EHR. The multi-level security supported by ML-ABE allows us to protect the relationships between EHR independently blocks with different levels of security, where only legitimate DU with appropriate attributes can decrypt a certain number of pointers and join the open data in a sensible way. These not only improve the security of EHR, but also grant patients the ability to share EHRs efficiently. In addition, with the formal security analysis, our solutions have been proven to be capable of preventing a range of possible security attacks. Finally, we have analysed the costs and performance of the E-Tenon system in various aspects. The simulation results show that our E-Tenon system does not compromise any security properties while maintaining promising efficiency and flexibility. In our future work, we aim to explore the use of Fog Computing technology to enhance our solutions in a way that can further improve efficiency (e.g., by outsourcing part of the EHR preprocessing and decryption tasks to edge devices) while maintaining strong security.

\ifCLASSOPTIONcaptionsoff
  \newpage
\fi



\bibliographystyle{IEEEtran}
\bibliography{reference}
%



%






\end{document}